\def\draft{1}
\def\llncs{0}

\ifnum\llncs=1
\documentclass[runningheads,a4paper]{llncs}
\usepackage{amssymb,amsmath,url,enumerate,latexsym,cite}
\usepackage{macrosCRYPTO}
\else
\documentclass[11pt]{article}
\usepackage{macros}
\usepackage{graphicx}
\usepackage{breakcites}
\usepackage{fullpage}
\usepackage[utf8]{inputenc}
\usepackage[T1]{fontenc}
\usepackage{libertine} 
\fi

\ifnum\llncs=1
\authorrunning{Gupte, Vaikuntanathan}
\titlerunning{How to Construct Quantum FHE, Generically}
\index{Gupte, Aparna}
\index{Vaikuntanathan, Vinod}
\fi

\title{How to Construct Quantum FHE, \\ Generically}
\ifnum\llncs=1
\author{Aparna Gupte \and Vinod Vaikuntanathan}
\institute{MIT CSAIL}

\else
\author{Aparna Gupte\thanks{Research supported by the Ida M. Green MIT Office of Graduate Education Fellowship and by the grants of the second author.}\\MIT\\
\texttt{agupte@mit.edu}\and Vinod Vaikuntanathan\thanks{Research supported in part by DARPA under Agreement Number HR00112020023, NSF CNS-2154149, a Simons Investigator award and a Thornton Family Faculty Research Innovation Fellowship
from MIT.}\\MIT\\\texttt{vinodv@mit.edu}}
\fi

\begin{document}
\maketitle
\begin{abstract}
    We construct a (compact) quantum fully homomorphic encryption (QFHE) scheme starting from {\em any} (compact) classical fully homomorphic encryption scheme with decryption in $\mathsf{NC}^{1}$, together with a dual-mode trapdoor function family. Compared to previous constructions (Mahadev, FOCS 2018; Brakerski, CRYPTO 2018) which made non-black-box use of similar underlying primitives, our construction provides a pathway to instantiations from different assumptions. Our construction uses the techniques of Dulek, Schaffner and Speelman (CRYPTO 2016) and shows how to make the client in their QFHE scheme classical using dual-mode trapdoor functions. As an additional contribution, we show a new instantiation of dual-mode trapdoor functions from group actions.
\end{abstract}

\ifnum\llncs=0
\newpage 
\tableofcontents
\newpage
\fi

\section{Introduction}

A quantum fully homomorphic encryption (QFHE)~\cite{mahadev2020classical} scheme enables a classical client with an input $x$ to outsource the computation of a quantum circuit $Q$ (on $x$) to an untrusted quantum server, while hiding $x$.  The first QFHE scheme was constructed under the (super-polynomial) learning with errors (LWE) assumption in a groundbreaking work of Mahadev~\cite{mahadev2020classical}; this was followed shortly after by a work of Brakerski~\cite{DBLP:conf/crypto/Brakerski18} who, with a different construction, improved the assumption to polynomial LWE.\footnote{In both cases, we get a leveled QFHE scheme, one that supports quantum circuits of an a-priori bounded depth, from LWE. To construct an unleveled QFHE scheme, both works rely on an appropriate circular security assumption, different from the one used to get an unleveled classical FHE scheme. From this point on, we focus on constructing a leveled QFHE scheme.}

The QFHE constructions of \cite{mahadev2020classical,DBLP:conf/crypto/Brakerski18} built on classical FHE schemes that can be constructed from the LWE assumption. However, their schemes use very specific properties of specific FHE schemes. For example, \cite{mahadev2020classical} used the dual-GSW FHE scheme~\cite{DBLP:conf/crypto/GentrySW13} together with a noisy trapdoor claw-free function built from LWE. The construction used the fact that the dual-GSW encryption of a bit $b$ can be converted into the description of a claw-free function pair $(f_0,f_1)$ that ``encodes'' the bit, in the sense that for two pre-images $x_0$ and $x_1$ such that $f_0(x_0)=f_1(x_1)$, the first bits of $x_0$ and $x_1$ xor to $b$. In a similar vein, \cite{DBLP:conf/crypto/Brakerski18} crucially used the GSW encryption scheme and the discrete Gaussian structure of the set of all random strings consistent with a given ciphertext. In other words, both schemes used a particular instantiation of the underlying primitives and exploited their intricate interplay.

This raises a natural question which is the starting point of our work: 
\begin{center}
    \begin{quote}
        {\em Can we generically transform any FHE scheme into a QFHE scheme?}
    \end{quote}
\end{center}

\noindent
Aside from the aesthetic appeal of such a result, it would give us a recipe to convert new classical FHE schemes, as they arise, into quantum FHE schemes; an example of a different FHE scheme is the IO-based FHE construction from \cite{canetti2015obfuscation}.

\subsection{Our Results} 
In this work, we make progress towards this goal by constructing a QFHE scheme starting with {\em any} (classical) FHE scheme (whose decryption circuit can be implemented in logarithmic depth) and {\em any} dual-mode trapdoor function (dTF) family. 

We define a dual-mode trapdoor function family as follows: a dTF is a pair of functions $(f_0,f_1)$ that either has the same image in one mode, or disjoint images in the other. Further, these two modes are computationally indistinguishable (see Section~\ref{section:dTFs} for the formal definition). In contrast, a plain claw-free trapdoor function (as defined in \cite{mahadev2018classical, brakerski2021cryptographic}) is a pair of functions $(f_0,f_1)$ for which it is computationally hard to find a claw, i.e. a pair of inputs $(x_0,x_1)$ such that $f_0(x_0)=f_1(x_1)$.

\begin{informaltheorem}\label{informaltheorem:main-theorem}
    Let $\secp$ be the security parameter. Assuming the existence of a leveled fully homomorphic encryption scheme $\mathsf{HE}$ whose decryption algorithm can be implemented as a Boolean circuit with depth $O(\log \secp)$, and a dual-mode trapdoor function family, there exists a quantum leveled fully homomorphic encryption scheme with a classical key generation procedure. The scheme encrypts classical information using classical ciphertexts.
\end{informaltheorem}

In other words, our result offers an answer to the open question above, except with two caveats: (1) we require the classical FHE scheme to have a shallow decryption circuit; and (2) we require an additional primitive, namely a dual-mode trapdoor function family. Both building blocks required for our construction can be instantiated from the Ring LWE assumption giving us the first construction of QFHE from Ring LWE. We note that one could have come up with this instantiation by looking into the innards of \cite{mahadev2020classical,DBLP:conf/crypto/Brakerski18}; however, with our construction, this instantiation is simply plug-and-play. Even with LWE-based constructions, any improvements in the underlying LWE-based FHE or LWE-based dual-mode trapdoor function translate directly into corresponding improvements in our QFHE scheme.

As an additional result, we show how to instantiate the dual-mode trapdoor function from group actions, building on the work of \cite{alamati2022candidate} who construct a plain claw-free trapdoor function from the same assumption. 

\begin{informaltheorem}
    There exists a dual-mode trapdoor function family under the extended linear hidden shift assumption (see \cite{alamati2022candidate} and Definition~\ref{def:elhs}) on group actions.
\end{informaltheorem}

It also turns out that we can construct an FHE scheme given an indistinguishability obfuscation (IO) scheme, together with a perfectly rerandomizable encryption scheme, following \cite{canetti2015obfuscation}. Together with a construction of a perfectly rerandomizable encryption scheme from group actions~\cite{Wichs}, this gives us a QFHE scheme from (post-quantum, subexponentially secure) IO and group actions.

\begin{informaltheorem}
    Assuming that post-quantum sub-exponentially secure indistinguishability obfuscation exists and assuming the post-quantum hardness of the extended linear hidden shift assumption, there is a quantum leveled fully homomorphic encryption scheme with a classical key generation procedure. The scheme encrypts classical information using classical ciphertexts.
\end{informaltheorem}

We remark that we can extend this result to give unleveled QFHE using post-quantum IO.

\subsection{Technical Overview}

The starting point of our QFHE construction is the beautiful work of Dulek, Schaffner and Speelman~\cite{dulek2016quantum} which predated \cite{mahadev2020classical}, and showed how to build a protocol whereby a {\em semi-quantum} client with a (classical or quantum) input $x$ can delegate to a fully quantum server with a quantum circuit $Q$ the computation of $Q(x)$. In addition to encrypting the input $x$, the client also generates polynomially many copies of a {\em quantum gadget} that helps the server do the quantum homomorphic evaluation of $Q$ on the encrypted $x$. 

In a nutshell, the key idea in our work is to come up with a mechanism for the client to outsource to the server the preparation of these quantum gadgets, using a dual-mode trapdoor function. This results in a completely classical client, and consequently a QFHE scheme.\footnote{This is, in effect, a remote state preparation protocol for the DSS gadget states; however, in the QFHE setting (as opposed to the verification setting), the protocol only needs to be secure against a semi-honest/specious server.} We now proceed to describe the construction in more detail, starting with the DSS protocol. 

The DSS scheme builds on the work of Broadbent and Jeffery~\cite{broadbent2015quantum} and encrypts quantum states (or classical strings) with a quantum one-time pad.
That is, a single-qubit state $\ket{\psi}$ is encrypted as 
$\X^{x}\Z^{z} \ket{\psi}$
where $x,z\gets \{0,1\}$ are uniformly random bits and 
$$ \X = \left( \begin{array}{cc} 0 & 1 \\ 1 & 0 
\end{array} \right),~
\Z = \left( \begin{array}{cc} 1 & 0 \\ 0 & -1 
\end{array} \right)
$$are the Pauli $\X$ and $\Z$ operators. Additionally, the ciphertext also contains a (classical) homomorphic encryption of the one-time pad keys $(x,z)$.

We recall that any quantum circuit can be implemented using Clifford gates and the single-qubit $\T$ gate 
$$ \T = \left( \begin{array}{cc} 1 & 0 \\ 0 & e^{i\pi/4} \end{array} \right)$$
Homomorphic evaluation of Clifford gates on states encrypted using the Pauli one-time pad is easy as Clifford gates nicely ``commute'' through the Pauli one-time pad. For example, 
$$ \H\X^x\Z^z \ket{\psi} = \Z^x\X^z\H\ket{\psi}$$
All that needs to be done is to homomorphically update encryption of the one-time pads.
When the non-Clifford $\T$ gate is applied to a one-time padded state, we end up with an extra non-Pauli error in the form of a $\P$ gate,
\begin{align*}
    \T \X^x \Z^z \ket{\psi} = \P^x \X^x \Z^z \T \ket{\psi}.
\end{align*}
The natural solution would be to apply a $\P^{\dagger}$ gate to this state if $x = 1$, and otherwise do nothing. However, the evaluator only has access to an encryption $\tilde{x}$ of the one-time pad key $x$. Dulek, Schaffner and Speelman~\cite{dulek2016quantum} solve this issue by writing the decryption algorithm (with the secret key $sk$ hardcoded) as a branching program and converting it into a  quantum ``teleportation gadget'' $\Gamma(sk)$ which the client prepares during key generation. We henceforth refer to this gadget as the DSS teleportation gadget. The server then teleports the qubit $\P^x \X^x \Z^z \T \ket{\psi}$ through the gadget in a way that depends on $\tilde{x}$, and the gadget implicitly decrypts $\tilde{x}$ and applies the $(\P^\dagger)^x$ correction to the qubit. The functionality achieved by the DSS gadget is depicted in Figure~\ref{fig:DSS}.
\ifnum\llncs=0
For completeness, we describe how the DSS gadget looks and how it is used to achieve the desired functionality in Appendix~\ref{section:DSSgadget}. 
\fi 

\begin{figure}[t]
\begin{center}
\includegraphics[height=2in]{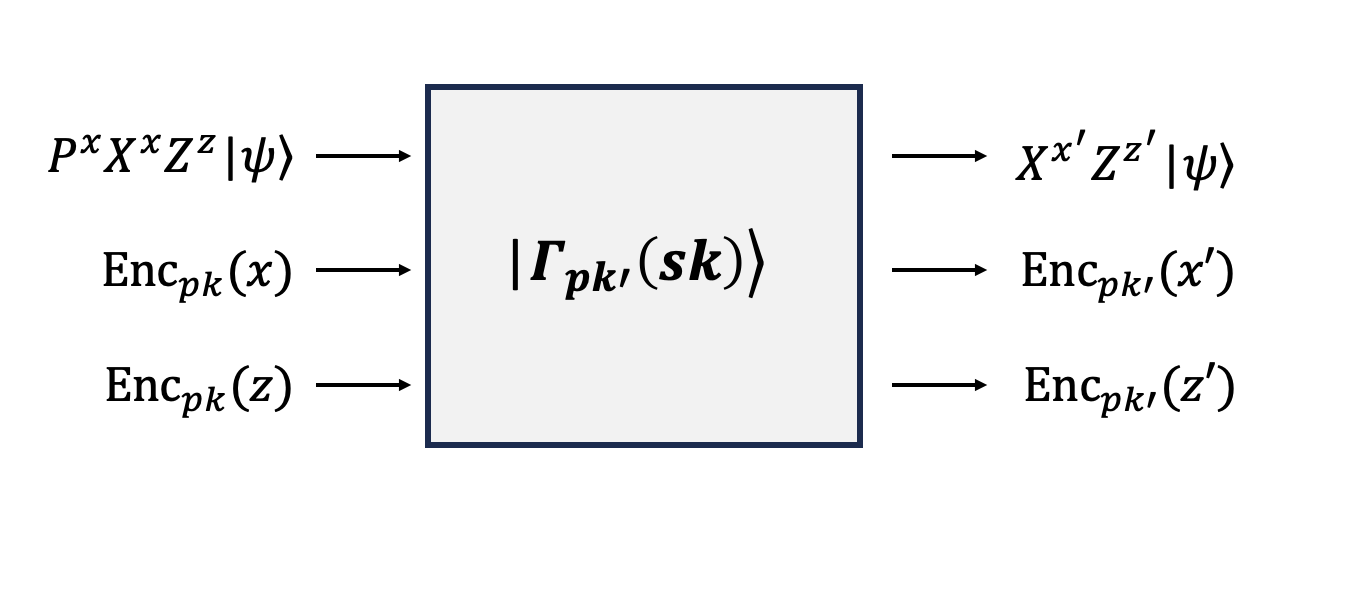}
\end{center}
\caption{The functionality of the DSS quantum gadget $\ket{\Gamma(sk)}$: take as input a state $P^xX^xZ^z\ket{\psi}$ as well as encryptions of $x$ and $z$, and produce as output $X^{x'}Z^{z'}\ket{\psi}$ together with encryptions of the new one-time pad keys $x'$ and $z'$.}
\label{fig:DSS}
\end{figure}

At a high level, the DSS gadget consists of Bell pairs that are wired in a way that depends on the bits of the secret key $sk$ used to decrypt the ciphertext $\Enc_{pk}(x)$. The key new contribution of our work is a method that allows a completely classical client to delegate the preparation of the DSS teleportation gadget to the server. At the center of our method is a procedure that prepares a ``hidden'' Bell pair. The procedure takes three qubits in registers $0, 1, 2$ in the product state $\ket{+}_0 \otimes \ket{+}_1 \otimes \ket{+}_2$, and a particular kind of ``encryption'' of a bit $\mu$. The procedure then outputs a (one-time padded) state where the qubits on registers $\mu$ and $2$ wired together to form a maximally entangled Bell state  that is in tensor product with the qubit $\ket{+}_{1-\mu}$ in the $1-\mu$ register,
\begin{align*}
    \ket{\Phi^+}_{\mu, 2} \otimes \ket{+}_{1-\mu}.
\end{align*}
We do this with the help of a ``$4$-to-$2$'' dual-mode trapdoor function (which we can obtain generically from dual-mode trapdoor functions, see Section~\ref{section:4to2dTFs}), a new primitive we define as follows (see Figure~\ref{figure:4-2-dTF} for a pictorial illustration): a tuple of four functions $f_{b_1, b_2}: \mathcal{X} \rightarrow \mathcal{Y}$ for $b_1, b_2$ such that depending on a ``mode'' bit $\mu \in\{0,1\}$, either
\begin{enumerate}
    \item{(Case $\mu = 0$)} two functions $f_{b_1, b_2}$ and $f_{b'_1, b'_2}$ have the same image if $b_1 = b'_1$ and otherwise they have disjoint images, or
    \item{(Case $\mu = 1$)} two functions $f_{b_1, b_2}$ and $f_{b'_1, b'_2}$ have the same image if $b_2 = b'_2$. Otherwise they have disjoint images.
\end{enumerate}
Further, without a trapdoor, no efficient adversary should be able to guess $\mu$ significantly better than randomly guessing.
Such a tuple of functions encoding the mode bit $\mu$ is the specific ``encryption'' we use to prepare the required hidden Bell pairs.
Given such a family of functions, we prepare the state
\begin{align*}
    \frac{1}{2 \sqrt{2 |\mathcal{X}|}}\sum_{\substack{u, v, w \in \{0,1\}\\x \in \mathcal{X}}} \ket{u, v, w}_{0, 1, 2} \otimes \ket{x}_3 \otimes \ket{f_{u \oplus w, v \oplus w}(x)}_4.
\end{align*}
Suppose for simplicity that $\mu = 0$.
Measuring the fourth register to get some $y \in \mathcal{Y}$ collapses the state to preimages under either the $b_1 = 0$ pair $(f_{0, b_2})_{b_2 \in \{0,1\}}$  or the $b_1 = 1$ pair $(f_{1, b_2})_{b_2 \in \{0,1\}}$. Again, for simplicity, let us consider the branch where $b_1 = 0$. The state we are now left with is
\begin{align*}
    \frac{1}{2} \left(\ket{000}\ket{x_{0,0}} + \ket{101}\ket{x_{0,1}} + 
    \ket{111}\ket{x_{0,0}} +
    \ket{010}\ket{x_{0,1}}\right)_{0,1,2,3},
\end{align*}
where $x_{b_1, b_2}$ is such that $f_{b_1, b_2}(x_{b_1, b_2}) = y$.
Now, we measure the register 3 (containing the preimages) in the Hadamard basis, so that we obtain a string $d$ and the state (up to global phase)
\begin{gather*}
    \frac{1}{2} \left( \ket{000} + (-1)^{z} \ket{101} + \ket{111} + (-1)^z \ket{010} \right) \\
    = \frac{1}{2} (\Z^{z} \otimes \Z^z \otimes \I) \ket{\Phi^+}_{0,2} \otimes \ket{+}_{1},
\end{gather*}
where $z = d \cdot (x_{0,0} \oplus x_{0,1})$.

This basic protocol allows us to dequantize the DSS client in the following manner: depending on the secret key $sk$, the client generates {\emph classical} evaluation keys for the $4$-to-$2$ dual-mode trapdoor functions (where the modes depend on the bits of $sk$) and sends these to the server. The server then performs the above procedure for all the qubits in the gadget, thereby wiring the Bell pairs correctly depending on the bits of $sk$. Once the server has prepared this gadget $\Gamma(sk)$, it can simply follow the DSS scheme for evaluating the $\T$ gates and correcting the $\P$ errors.

\subsection{Related Work}
\paragraph{Quantum Fully-Homomorphic Encryption.} Broadbent and Jeffery~\cite{broadbent2015quantum} gave the first QFHE scheme assuming only a classical FHE scheme, although they considered the setting of a quantum client. They constructed a (non-compact) scheme where the ciphertext size grows polynomially with the number of non-Clifford gates in the circuit. They also gave a different quantum FHE scheme in which the quantum client produces quantum evaluation keys, the size of which grows doubly exponentially in the non-Clifford gate-depth. Dulek, Schaffner and Speelman~\cite{dulek2016quantum} improved on this and gave a QFHE scheme, also for a quantum client, but with the size of the quantum evaluation key growing polynomially with the non-Clifford gate-depth. However, they required the underlying classical FHE scheme to have a decryption circuit in $\mathsf{NC}^1$.

The work of Mahadev~\cite{mahadev2020classical} was a major breakthrough and constructed the first QFHE scheme for a completely classical client, under the quantum hardness of LWE (with superpolynomial modulus-to-noise ratio). Brakerski~\cite{DBLP:conf/crypto/Brakerski18} later improved the parameters, and gave a different QFHE scheme for a classical client under a weaker assumption, namely LWE with a polynomial modulus-to-noise ratio.

\paragraph{Remote State Preparation.} The main conceptual idea of our work is to give a protocol by which the server can blindly prepare the quantum evaluation keys of \cite{dulek2016quantum}, given only classical messages from the client and therefore side-stepping the need for a quantum client as well as quantum communication. This idea is inspired by the work on remote state preparation, a protocol that allows a quantum server to blindly prepare a quantum state specified by the classical client. The notion of remote state preparation was first introduced by~\cite{dunjko2016blind}, and subsequently studied by \cite{gheorghiu2019computationally, gheorghiu2022quantum}. We note that remote state preparation usually refers to a stronger variant in the setting of a malicious server, where the protocol to prepare the quantum state is not only \textit{blind}, but also \textit{verifiable} by the classical client. However in the QFHE setting, the server is semi-honest, and we only require the weaker notion of blind (but not necessarily verifiable) remote state preparation.

\paragraph{Organization of the Paper.} In Section~\ref{section:preliminaries}, we introduce some notation and define (quantum) homomorphic encryption. In Section~\ref{section:dTFs}, we define dual-mode trapdoor functions (dTFs), and show how to construct a variant we call ``$4$-to-$2$'' dual-mode trapdoor functions from plain dual-mode trapdoor functions. We also give an amplification lemma that will later allow us to construct dTFs from group actions with negligible correctness error. In Section~\ref{section:dss-rsp} we use this primitive to show that we can dequantize the DSS client, and we present our QFHE scheme in Section~\ref{section:our-scheme}. Section~\ref{section:instantiations} lists the known instantiations of the two primitives we require, namely classical FHE and dual-mode trapdoor functions, from several cryptographic assumptions. Finally, in Section~\ref{sec:groupactions}, we show an instantiation based on hardness assumptions related to group actions.

\section{Preliminaries}\label{section:preliminaries}
\paragraph{Notation.}
Following the convention in \cite{dulek2016quantum}, we  denote ciphertexts with a tilde sign, when the encryption scheme and the public/private key are clear from the context. For example $\tilde{a}$ denotes an encryption of $a$. When a sequence of keys is involved, we write $\tilde{a}^{[i]}$ to denote an encryption of $a$ under the $i$-th key in a set.  If $\mathcal{M}$ is a quantum register, we denote the set of density operators on $\mathcal{M}$ by $D(\mathcal{M})$. We denote the trace distance between density matrices $\rho$ and $\sigma$ as $T(\rho, \sigma)$, and the total variation distance between probability distribution $D_1, D_2$ as $\TV(D_1, D_2)$.

\begin{definition}[Hellinger Distance] For two probability density functions $f, g$ over a finite domain $\mathcal{X}$, the squared Hellinger distance between $f$ and $g$ is defined as
\begin{align*}
    H^2(f, g) = \frac{1}{2} \sum_{x \in \mathcal{X}} \left(\sqrt{f(x)} - \sqrt{g(x)}\right)^2 = 1 - \sum_{x \in \mathcal{X}} \sqrt{f(x) g(x)}.
\end{align*}
\end{definition}

\begin{lemma}\label{lemma:hellinger-TV}
    Let $D_1, D_2$ be two probability density functions over a finite domain $\mathcal{X}$, then
    \begin{align*}
        \frac{1}{2}H^2(D_1, D_2) \le \TV(D_1, D_2) \le H(D_1, D_2).
    \end{align*}
\end{lemma}

\begin{lemma}\label{lemma:hellinger-trace-distance}
    Let $\mathcal{X}$ be a finite set and let $D_1, D_2$ be probability densities on $\mathcal{X}$. Let
    \begin{align*}
        \ket{\psi_1} = \sum_{x\in \mathcal{X}} \sqrt{D_1(x)} \ket{x} \quad \text{ and } \quad \ket{\psi_2} = \sum_{x \in \mathcal{X}} \sqrt{D_2(x)}\ket{x}\;.
    \end{align*}
    Then, the trace distance between the pure states $\ket{\psi_1}, \ket{\psi_2}$ is
    \begin{align*}
        T(\ket{\psi_1} \bra{\psi_1}, \ket{\psi_2} \bra{\psi_2}) = \sqrt{1 -  (1 - H^2(D_1, D_2))^2}.
    \end{align*}
\end{lemma}

\subsection{Homomorphic Encryption}
\begin{definition}[Classical Homomorphic encryption scheme] Let $\secp$ be the security parameter. A homomorphic encryption scheme $\HE$ is a tuple of p.p.t.~ algorithms $\HE = (\HE.\KeyGen, \HE.\Enc, \HE.\Dec, \HE.\Eval)$ with the following properties.
\begin{itemize}
    \item{Key Generation.} The probabilistic key generation algorithm $(pk, evk, sk) \leftarrow \HE.\KeyGen(1^\secp)$ outputs a public key $pk$, an evaluation $evk$ and a secret key $sk$.
    \item{Encryption.} The probabilistic encryption algorithm $c \leftarrow \HE.\Enc_{pk}(x)$ uses the public key $pk$ and encrypts a single bit message $x \in \{0,1\}$ into a ciphertext $c$.
    \item{Decryption.} The deterministic decryption algorithm $x \leftarrow \HE.\Dec_{sk}(c)$ uses the secret key $sk$ and decrypts a ciphertext $c$ to recover the message $x \in \{0,1\}$.
    \item{Evaluation.} The (deterministic or probabilistic) evaluation algorithm $c' \leftarrow \HE.\Eval_{evk}(\mathcal{C}, (c_1, \ldots, c_\ell))$ uses the evaluation key $evk$, takes as input a circuit $\mathcal{C}:\{0,1\}^\ell \rightarrow\{0,1\}$ and sequence of ciphertexts $c_1, \ldots, c_\ell$, and outputs a ciphertext $c'$.
\end{itemize}
\end{definition}
We often overload the functionality of the encryption and decryption procedures by allowing them to take in multi-bit messages as inputs, and produce a sequence of ciphertexts that correspond to bit-by-bit encryptions.

\begin{definition}[Full homomorphism and compactness]
    A scheme $\HE$ is fully homomorphic, if for any efficiently computable circuit $\mathcal{C}$ and any set of inputs $x_1, x_2, \ldots, x_\ell$,
    \begin{align*}
        \Pr[\HE.\Dec_{sk}(c') \neq \mathcal{C}(x_1, \ldots, x_\ell)] = \negl(\secp).
    \end{align*}
where the probability is taken over $(pk, evk, sk) \leftarrow \HE.\KeyGen(1^\secp)$ and $c_i \leftarrow \HE.\Enc_{pk}(x_i)$, and $c' \leftarrow \HE.\Eval_{evk}(\mathcal{C}, (c_1, \ldots, c_\ell))$.

A fully homomorphic encryption scheme is compact if its decryption circuit is independent of the evaluated function. The scheme is leveled fully homomorphic if it takes $1^L$ as additional input in key generation, and can only evaluate depth $L$ Boolean circuits.
\end{definition}

\subsection{Quantum Homomorphic Encryption}
We now give a definition for quantum homomorphic encryption based on the definition given in \cite{broadbent2015quantum}.
In this work, we focus on quantum homomorphic encryption schemes with entirely classical key generation procedures.
\begin{definition}[Quantum Homomorphic Encryption]
A quantum homomorphic encryptions scheme is a $\mathsf{QHE}$ is a tuple of quantum polynomial-time algorithms $(\mathsf{QHE}.\mathsf{KeyGen}, \mathsf{QHE}.\mathsf{Eval}, \mathsf{QHE}.\mathsf{Dec})$:
\begin{itemize}
    \item{Key Generation.} The probabilistic key generation algorithm $(pk, sk, evk) \leftarrow \mathsf{QHE}.\mathsf{KeyGen}(1^\lambda)$ takes a unary representation of the security parameter as input and outputs a classical public key $pk$, a classical secret key $sk$ and a classical evaluation key $evk$.
    \item{Encryption.} For every possible value of $pk$, the quantum channel $\mathsf{QHE}.\mathsf{Enc}_{pk} : D(\mathcal{M}) \rightarrow D(\mathcal{C})$ maps a state in the message space $\mathcal{M}$ to a state (the cipherstate) in the cipherspace $\mathcal{C}$.
    \item{Decryption.} For every possible value of $sk$, $\mathsf{QHE}.\mathsf{Dec}_{sk} : D(\mathcal{C}') \rightarrow \mathcal{M}$ is a quantum channel that maps the state in $D(\mathcal{C}')$ to a quantum state in $D(\mathcal{M})$.
    \item{Homomorphic Circuit Evaluation.} For every quantum circuit $\mathsf{C}$, with induced channel $\Phi_{\mathsf{C}}: D(\mathcal{M}^{\otimes n}) \rightarrow D(\mathcal{M}^{\otimes m})$, we define a channel $\mathsf{QHE}.\mathsf{Eval}_{evk}(\mathsf{C}, \cdot) : D(\mathcal{C}^{\otimes n}) \rightarrow D(\mathcal{C}'^{\otimes m})$ that maps an $n$-fold cipherstate to an $m$-fold cipherstate.
\end{itemize}
\end{definition}

\begin{definition}[Quantum Full Homomorphism and Compactness] Let $\lambda$ be the security parameter. A quantum homomorphic encryption scheme $\QHE$ is fully homomorphic, if for any efficiently computable quantum circuit $\mathsf{C}$ with induced channels $\Phi_\mathsf{C} : \mathcal{M}^{\otimes n(\lambda)} \rightarrow \mathcal{M}^{\otimes m(\lambda)}$, and for any input $\rho \in D(\mathcal{M}^{\otimes n(\lambda)} \otimes \mathcal{E})$, there exists a negligible function $\negl$ s.t. for $(pk, sk, evk) \leftarrow \QHE.\KeyGen(1^\lambda)$, the state
\begin{align*}
    \QHE.\Dec_{sk}^{\otimes m(\lambda)} \left( \QHE.\Eval_{evk}(\mathsf{C}, \QHE.\Enc_{pk}^{n(\lambda)}(\rho)) \right)
\end{align*}
is at most $\negl(\lambda)$-away in trace distance from the state $\Phi_{\mathsf{C}}(\rho)$.

A quantum homomorphic encryption is compact if its decryption circuit is independent of the evaluated circuit.

The scheme is leveled fully homomorphic if it takes $1^L$ as additional input in key generation, and can only evaluate circuits of $\T$-depth at most $L$.
\end{definition}

\def\QINDCPA{\mathsf{G}^{\text{q-CPA}}}

\begin{definition}[The Quantum IND-CPA Game~\cite{broadbent2015quantum}] Let $\mathcal{M}, \mathcal{S}, \mathcal{C}$ be registers. The quantum CPA indistinguishability experiment with respect to a quantum homomorphic encryption scheme $\mathsf{S}$ and a quantum polynomial-time adversary $\mathcal{A} = (\mathcal{A}_1, \mathcal{A}_2)$ with security parameter $\lambda$, denoted as $\QINDCPA_{\langle\mathsf{S}, \mathcal{A}\rangle}(1^\lambda)$ is defined as follows.
\begin{enumerate}
    \item The challenger runs $\mathsf{S}.\mathsf{Gen}(1^\lambda)$ to obtain keys $(pk, sk, evk)$.
    \item Adversary $\mathcal{A}_1$ is given $(pk, evk)$ and outputs a quantum state $\rho_{\mathcal{M}, \mathcal{S}}$ on $\mathcal{M} \otimes \mathcal{S}$.
    \item The challenger samples a uniformly random bit $r \leftarrow \{0,1\}$. If $r = 0$, the challenger outputs in $\mathcal{C}$ an encryption $\mathsf{S}.\mathsf{Enc}_{pk}(\ket{0} \bra{0})$ of the state $\ket{0} \bra{0}$. If $r = 1$, the challenger outputs in $\mathcal{C}$ an encryption $\mathsf{S}.\mathsf{Enc}_{pk}(\rho_{\mathcal{M}})$ of the state in $\mathcal{M}$.
    \item Adversary $\mathcal{A}_2$ obtains $\mathcal{C} \otimes \mathcal{S}$ and outputs a bit $r'$.
    \item If $r = r'$, the adversary $\mathcal{A}$ wins, and the outcome of the experiment is defined to be $1$. Otherwise, the outcome of the experiment is defined to be $1$.
\end{enumerate}
\end{definition}

\begin{definition}[Quantum CPA Indistinguishability~\cite{broadbent2015quantum}]
    A quantum homomorphic encryption scheme $\mathsf{S}$ is quantum-IND-CPA secure if there is no quantum polynomial-time adversary wins the quantum IND-CPA game with more than negligible advantage. That is, for any quantum polynomial-time adversary $\mathcal{A} = (\mathcal{A}_1, \mathcal{A}_2)$, there exists a negligible function $\negl$ such that
    \begin{align*}
        \Pr[\QINDCPA_{\langle \mathsf{S}, \mathcal{A} \rangle}(1^\lambda) = 1] \le \frac{1}{2} + \negl(\lambda).
    \end{align*}
\end{definition}
Broadbent and Jeffery~\cite{broadbent2015quantum} show that the above security definition is equivalent to security w.r.t.~a quantum IND-CPA game where the adversary sends multiple messages to the challenger, and the challenger either encrypts all of them or encrypts a sequence of all $\ket{0}\bra{0} $ states.

\ifnum\llncs=0
\subsection{Discrete Gaussians}

\begin{definition}[Truncated Discrete Gaussian] Let $B > 0$ be a real number and let $q$ be a positive integer. The truncated discrete Gaussian distribution over $\mathbb{Z}_q$ with parameter $B$ is the distribution supported on $\{x \in \mathbb{Z}_q : \|x\| \le B\}$ with density proportional to
\begin{align*}
    \rho_{\mathbb{Z}_q, B}(x) \propto \exp\left( \frac{-\pi \|x\|^2}{B^2} \right) \;.
\end{align*}
For a positive integer $m$ the truncated discrete Gaussian distribution on $\mathbb{Z}_q^m$ with parameter $B$ is the distribution supported on $\{\vecx \in \mathbb{Z}_q : \|x\| \le B \sqrt{m}\}$ with density
\begin{align*}
     \rho_{\mathbb{Z}^m_q, B}(\vecx) = \rho_{\mathbb{Z}_q, B}(x_1) \cdots \rho_{\mathbb{Z}_q, B}(x_m).
\end{align*}
\end{definition}

\begin{lemma}\label{lemma:hellinger-shifted-discrete-gaussian}
Let $B > 0$ be a real number and let $q, m$ be positive integers. Consider $\vece \in \mathbb{Z}_q^m$ such that $\|\vece\| \le B\sqrt{m}$. The Hellinger distance between the distribution $\rho = \rho_{\mathbb{Z}^m_q, B}$ and the shifted distribution $\rho + \vece$ with density $(\rho + \vece)(\vecx) = \rho(\vecx - \vece)$ satisfies
\begin{align*}
    H^2(\rho, \rho + \vece) \le 1 - \exp \left( \frac{-2 \pi \sqrt{m} \|\vece\|}{B} \right) \le \frac{2 \pi \sqrt{m} \|\vece\|}{B} \;.
\end{align*}
\end{lemma}

\subsection{Learning With Errors}
\begin{definition}[The $\LWE_{m, n, q, \chi}$ problem] Let $\lambda$ be the security parameter, let $m, n, q$ be integer functions of $\lambda$. Let $\chi = \chi(\lambda)$ be a distribution over $\mathbb{Z}$. The $\LWE_{m, n, q, \chi}$ problem is to distinguish between the distributions $(\matA, \matA \vecs + \vece \pmod{q})$ and $(\matA, \vecu)$, where $\matA \leftarrow \mathbb{Z}^{m \times n}_q$, $\vecs \leftarrow \mathbb{Z}_q^n$, $\vece \leftarrow \chi^m$ and $\vecu \leftarrow \mathbb{Z}_q^m$.

The problem where $m$ can be an arbitrary polynomially bounded function of $n \log q$ is written as $\LWE_{n, q, \chi}$.
\end{definition}

\begin{definition}[The $\LWE_{n, q, \chi}$ Assumption] There is no quantum polynomial-time proceudre that solves the $\LWE_{n, q, \chi}$ problem with advantage more than negligible in $\lambda$.
\end{definition}

There is evidence to believe that the $\LWE_{n, q, \rho_{\mathbb{Z}_q, \sigma}}$ assumption is true for certain parameter regimes, where $\sigma = \alpha q \ge 2\sqrt{n}$, and $\alpha$ is the noise-to-modulus ratio. In particular, \cite{regev2009lattices} and \cite{peikert2017pseudorandomness} show that a quantum polynomial-time algorithm for $\LWE_{n, q, \rho_{\mathbb{Z}_q, \sigma}}$ would give a quantum polynomial-time algorithm for SIVP (shortest independent vectors problem) in for $n$-dimensional lattices within an approximation factor of $\gamma = \tilde{O}(n/\alpha)$. The best known (classical or quantum) algorithms for SIVP run in time $2^{\tilde{O}(n/\log \gamma)}$, which is efficient for $\gamma = 2^{\tilde{O}(n/\log n)}$. If SIVP is hard for a (smaller) superpolynomial approximation factor $\gamma$, $\LWE$ is hard for a corresponding inverse superpolynomial noise-to-modulus ratio.

\begin{theorem}[Theorem~5.1 in \cite{micciancio2012trapdoors}]\label{theorem:mp12}
Let $n, m \ge 1$ and $q \ge 2$ be such that $m = \Omega(n \log q)$. There is an efficient randomized algorithm $\mathsf{GenTrap}(1^n, t^m, q)$ that returns a matrix $\matA \in \mathbb{Z}^{m \times n}_q$ and a trapdoor $t_\matA$ such that the distribution of $\matA$ is negligibly (in $n$) close to the uniform distribution. Moreover, there is an efficient algorithm $\mathsf{Invert}$ that, on input $\matA, t_\matA$ and $\matA \vecs + \vece$ where $\|\vece\|_2 \le q/(C \sqrt{n \log q})$ and $C$ is a universal constant, returns $\vecs$ and $\vece$ with overwhelming probability over $(\matA, t_\matA) \leftarrow \mathsf{GenTrap}(1^n, t^m, q)$.
\end{theorem}
\fi

\section{Dual-mode Trapdoor Functions}\label{section:dTFs}
A central primitive we make use of in our QFHE scheme is what we will call a dual-mode trapdoor function (dTF) family (and a variant, $4$-to-$2$ dTFs, defined in Section~\ref{section:4to2dTFs}). This primitive is a slight variant of noisy trapdoor claw-free function families (NTCFs) and \textit{extended} NTCFs, which find several applications in quantum cryptography, starting with the breakthrough schemes for QFHE and classical verification of BQP computations of Mahadev~\cite{mahadev2020classical, mahadev2018classical}.\footnote{Our dTFs are a weakening of Mahadev's extended NTCFs in two ways: we do not require the adaptive hardcore bit property, and we do not need the functions to be injective. Relaxing the definition to allow many-to-one functions allows us to obtain dTFs with good correctness properties from group actions, as well as dTFs from weaker LWE assumption based on \cite{DBLP:conf/crypto/Brakerski18}.}

Informally, a dTF family is a collection of pairs of functions $f_0, f_1$ that have either disjoint images (disjoint mode $\mu = 0$) or they the same image (lossy mode $\mu = 1$). We note that in reality the lossy mode property is slightly more technical to accommodate non-injective functions. %
Further, given the evaluation keys for $f_0, f_1$, it should be computationally hard to tell in which mode it is. The keys are sampled along with trapdoors that allow efficient inversion, and therefore also distinguish the two modes. See Figure~\ref{figure:dTF} for intuition.

Section~\ref{section:instantiations} describes constructions of dTFs from LWE and group actions. We note that none of these constructions give us ``clean'' dTFs; each assumption gives rise to a different notion of noise/error that deviates from the natural ideal definition. To capture both these different notions of noise, we introduce a probability distribution over the domain $\mathcal{X}$, and an correctness-error parameter $\epsilon$ in our definition. Although the construction of dTFs from group actions gives us inverse-polynomial correctness error at first, we show an amplification lemma in Section~\ref{section:dTF-amplification} that allows us to obtain negligible correctness error.


\begin{figure}[t]
    \centering
\begin{subfigure}{0.30\textwidth}
    \begin{tikzpicture}
        \draw (-1, 0) ellipse (0.5 and 1) node[above](X1){} node[below](X2){};
        \node at (-1,1.3){$\mathcal{X}$};
        \draw (2, 0) ellipse (1 and 2);
        \node at (2, 2.3){$\mathcal{Y}$};
        \draw (2, 0.7) circle (0.5) node[left](Y1){};
        \draw (2, -0.7) circle (0.5) node[left](Y2){};
        \draw[->, bend left=20] (X1) to node[midway, above] {$f_{0}$} (Y1);
        \draw[->, bend left=-10] (X2) to node[midway, below] {$f_{1}$} (Y2);
    \end{tikzpicture}
    \subcaption{$\mu = 0$}
\end{subfigure}\hspace{0.2\textwidth}
\begin{subfigure}{0.30\textwidth}
    \begin{tikzpicture}
        \draw (-1, 0) ellipse (0.5 and 1) node[above](X1){} node[below](X2){};
        \node at (-1,1.3){$\mathcal{X}$};
        \draw (2, 0) ellipse (1 and 2);
        \node at (2, 2.3){$\mathcal{Y}$};å
        \draw (2, 0) ellipse (0.5 and 0.7) node[above](Y1){} node[below](Y2){};
        \draw[->, bend left=20] (X1) to node[midway, above] {$f_{0}$} (Y1);
        \draw[->, bend left=-20] (X2) to node[midway, below] {$f_{1}$} (Y2);
    \end{tikzpicture}
    \subcaption{$\mu = 1$}
\end{subfigure}
\caption{Dual-mode trapdoor functions (dTFs): In mode $\mu = 0$, the functions $f_{0}$ and $f_{1}$ have disjoint images, and in mode $\mu = 1$ they have the same image. The two modes are computationally indistinguishable.\label{figure:dTF}}
\end{figure}
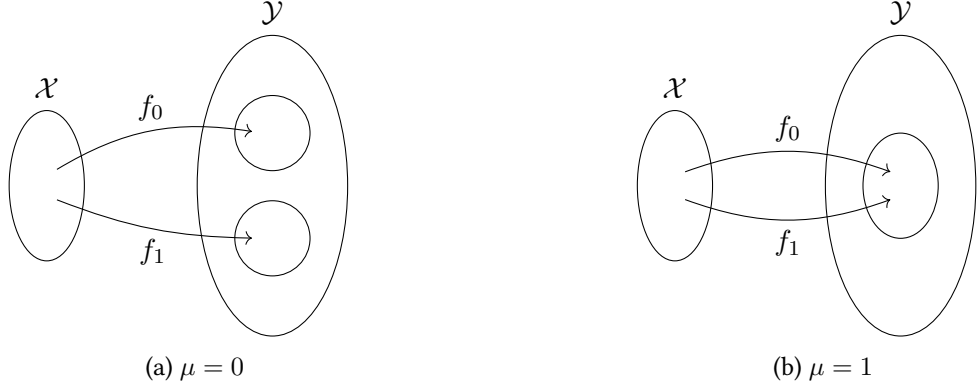

\begin{definition}[$\epsilon$-Weak Dual-Mode Trapdoor Function Family]\label{definition:weak-lossy-function}
Let $\lambda$ be the security parameter, and let $\mathcal{X}$, $\mathcal{Y}$ be finite sets.\footnote{We assume that there is an efficient binary encoding that maps elements $x \in \mathcal{X}$ to binary strings of length $t$, say. For simplicity, we do not use additional notation to distinguish between a set element $x\in \mathcal{X}$ and its binary representation. We write $\bar{\mathcal{X}}$ to denote $\{0,1\}^t$, and when we write $d \cdot x$ for $d \in \bar{\mathcal{X}}$, we think of $x$ as its binary encoding, and the dot product as being over $\mathbb{F}_2$.} Let $\mathcal{K}$ be a finite set of keys. A family of functions
\begin{align*}
    \mathcal{F} = \{f_{k, b} : \mathcal{X} \rightarrow \mathcal{Y}\}_{k \in \mathcal{K}, b \in \{0,1\}}
\end{align*}
is called a \textbf{$\epsilon$-weak dual-mode trapdoor function family} if there exists an associated a family of probability distributions $\mathcal{D} = \{ D_{k, b} \}_{k \in \mathcal{K}, b \in \{0,1\}}$ over $\mathcal{X}$ such that the following conditions hold.
\begin{enumerate}
    \item{\textbf{Efficient key generation.}} There is an efficient probabilistic algorithm $\mathsf{Gen}_\mathcal{F}$ which takes as input the security parameter $1^\lambda$ and a ``mode'' bit $\mu \in\{0,1\}$, and generates a key $k(\mu) \in \mathcal{K}$ together with a trapdoor $t_k$,
    \begin{align*}
        (k(\mu), t_k) \leftarrow \mathsf{Gen}_{\mathcal{F}}(1^\secp, \mu).
    \end{align*}
    \item{\textbf{Efficient function evaluation.}} There is an efficient deterministic algorithm that takes as input any $k\in\mathcal{K}$, $b \in \{0,1\}$, and $x\in\mathcal{X}$ and computes $f_{k, b}(x)$.
    \item{\textbf{Efficient state preparation.}} There is a quantum polynomial-time algorithm that prepares, for every $k\in \mathcal{K}, b \in \{0,1\}$, a state negligibly close in trace distance to the following state:
    \begin{align*}
        \sum_{x \in \mathcal{X}} \sqrt{D_{k, b}(x)} \ket{x}.
    \end{align*}
    \item{\textbf{Efficient Partial Inversion and Phase Computation with the Trapdoor.}}\label{property:efficient-phase-computation} There exists an efficient deterministic algorithm that takes as input $t_k, y$ and outputs the set 
    \begin{align*}
        \{b \in \{0,1\} : \exists x \in \mathcal{X} \text{ s.t. } f_{k, b}(x) = y\}.
    \end{align*}
    
    Additionally, there exists another efficient deterministic algorithm that takes as input $y \in \mathcal{Y}$, $d\in \bar{\mathcal{X}}$ and $b \in \{0,1\}$, along with trapdoor $t_k$, and computes (up to a common normalization factor) $\alpha_{k, y, d}(b)$, defined as
    \begin{align*}
        \alpha_{k, y, d}(b) := \sum_{x \in \mathcal{X}} (-1)^{d \cdot x} \sqrt{D_{k, b}(x)}\;.
    \end{align*}
    When $k$, $y$ and $d$ are clear from the context, we will sometimes abbreviate $\alpha_{k, y, d}(b)$ as $\alpha_b$.
    \item{\textbf{Dual-mode.}} Intuitively, we want the images of functions $f_{k, 0}$ and $f_{k, 1}$ to be either disjoint ($\mu=0$) or nearly identical ($\mu = 1$), depending on the mode $\mu$. More precisely, for any $\mu \in \{0,1\}$, for all but a negligible fraction of keys $(k(\mu), \cdot) \leftarrow \mathsf{Gen}_{\mathcal{F}}(1^\lambda, \mu)$,
    \begin{enumerate}
        \item\label{item:disjoint-mode} Disjoint mode ($\mu = 0$): For $x \leftarrow D_{k, 0}$, $x' \leftarrow D_{k, 1}$ we have that $f_{k, 0}(x) \neq f_{k, 1}(x')$ with probability $1$.
        \item\label{item:lossy-mode} Lossy mode ($\mu = 1$): 
        for $b\leftarrow \{0,1\}$, $x\leftarrow D_{k, b}$, $y = f_{k, b}(x)$, with probability $1-\epsilon$ there exists some $s \in \{0,1\}$ such that $\alpha_{k, y, d}(0) = (-1)^s \cdot \alpha_{k, y, d}(1)$ for all $d \in \bar{\mathcal{X}}$.
    \end{enumerate}
    \item{\textbf{Mode Indistinguishability.}} Given a key $k(\mu)$ generated with $\mathsf{Gen}(1^\lambda, \mu)$, no efficient algorithm can guess $\mu$ with greater than negligible advantage. Formally, for all quantum polynomial-time procedures $\mathcal{A}$, there exists a negligible function $\negl(\cdot)$ such that
    \begin{align*}
        \left| \Pr_{(k, t_k) \leftarrow \mathsf{Gen}_{\mathcal{F}}(1^\secp, 0)}[\mathcal{A}(k)] - \Pr_{(k, t_k) \leftarrow \mathsf{Gen}_{\mathcal{F}}(1^\secp, 1)}[\mathcal{A}(k)] \right| \le \negl(\lambda).
    \end{align*}
\end{enumerate}
\end{definition}

\begin{remark}
    If $\epsilon$ is negligible in the security parameter, we simply say the function family is a dual-mode trapdoor function family.
\end{remark}
For our amplification lemma (Section~\ref{section:dTF-amplification}), it will be convenient for us to define a stronger notion of dTFs, where the functions are injective and efficiently invertible with a trapdoor.

\begin{definition}\label{definition:injective-and-invertible}
We additionally call an $\epsilon$-weak dTF family $\mathcal{F}$ \textbf{injective and invertible} if the following properties hold respectively.
\begin{enumerate}
    \item For all $k \in \mathcal{K}$, $b \in \{0,1\}$, $f_{k, b}$ is injective on the support of $D_{k, b}$. That is, for all $x \neq x' \in \support(D_{k, b})$, we have that $f_{k,b}(x) \neq f_{k,b}(x')$.
    \item There exists an efficient deterministic algorithm $\mathsf{Inv}_{\mathcal{F}}$ such that for all $y \in \mathcal{Y}$ and key-trapdoor pairs $(k, t_k)$ generated by $\mathsf{Gen}_{\mathcal{F}}$,
    \begin{align*}
        \{(b, x) \mid b \in \{0,1\}, x \in \mathcal{X}, y = f_{k, b}(x)\} \leftarrow \mathsf{Inv}_{\mathcal{F}}(t_k, y).
    \end{align*}
    \item For every $b, b' \in \{0,1\}$ and $x, x'\in \mathcal{X}$, $f_{k, b}(x) = f_{k, b'}(x')$ implies that $D_{k, b}(x) = D_{k, b'}(x')$.
\end{enumerate}
\end{definition}

\subsection{$4$-to-$2$ dTFs and their construction from dTFs}\label{section:4to2dTFs}
In this subsection, we define $4$-to-$2$ dTFs and give a simple transformation that takes a dTF family and outputs a $4$-to-$2$ dTF family. We define a $4$-to-$2$ dual-mode trapdoor function to be a family of $4$-tuples of functions $f_{b_1, b_2}$ indexed by $b_1, b_2 \in \{0,1\}$. In mode $\mu = 0$, functions $f_{0,0}$ and $f_{0,1}$ have the same image, as do functions $f_{1,0}$ and $f_{1,1}$. Further, these two images are disjoint. Symmetrically, in mode $\mu = 1$, functions $f_{0,0}$ and $f_{1,0}$ share an image, as do functions $f_{0, 1}$ and $f_{1,1}$. The intuition for this primitive is best conveyed through a picture (see Figure~\ref{figure:4-2-dTF}).\footnote{We remark that there is a natural way to generalize dTFs and $4$-to-$2$ dTFs so that they are specializations of the same primitive. In short, we can define a $k$-mode trapdoor function family to be a collection of tuples of injective functions $(f_i)_{i \in \mathcal{I}}$ with finite index set $\mathcal{I}$. The two modes are defined by two distinct partition functions $p_0, p_1:\mathcal{I}\rightarrow \{0,1\}$ of $\mathcal{I}$. In mode $\mu \in [k]$ for some $k\in\mathbb{N}$, two functions $f_i, f_{i'}$ have the same image if $p_\mu(i) = p_{\mu}(i')$; otherwise $f_i, f_{i'}$ have disjoint images. We conjecture that this generalized notion could be useful in other applications. %
However, we present separate definitions of both dTFs and $4$-to-$2$ dTFs for the sake of clarity, at the cost of some redundancy.}

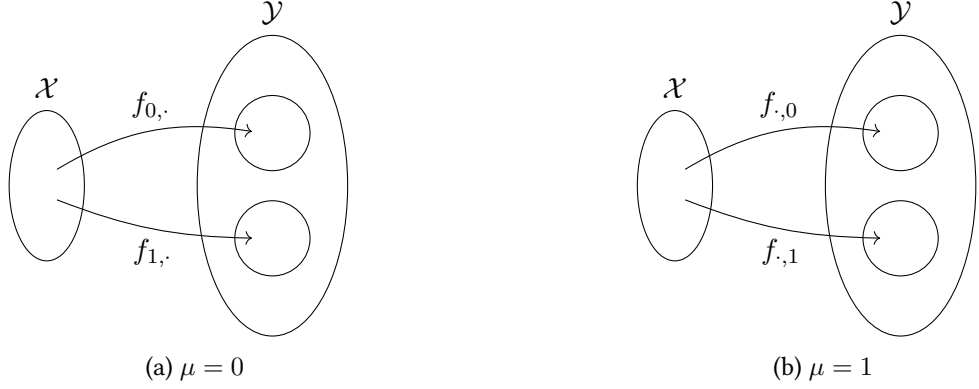
\begin{figure}[t]
    \centering
\begin{subfigure}{0.30\textwidth}
    \begin{tikzpicture}
        \draw (-1, 0) ellipse (0.5 and 1) node[above](X1){} node[below](X2){};
        \node at (-1,1.3){$\mathcal{X}$};
        \draw (2, 0) ellipse (1 and 2);
        \node at (2, 2.3){$\mathcal{Y}$};
        \draw (2, 0.7) circle (0.5) node[left](Y1){};
        \draw (2, -0.7) circle (0.5) node[left](Y2){};
        \draw[->, bend left=20] (X1) to node[midway, above] {$f_{0, \cdot}$} (Y1);
        \draw[->, bend left=-10] (X2) to node[midway, below] {$f_{1, \cdot}$} (Y2);
    \end{tikzpicture}
    \subcaption{$\mu = 0$}
\end{subfigure}\hspace{0.2\textwidth}
\begin{subfigure}{0.30\textwidth}
    \begin{tikzpicture}
        \draw (-1, 0) ellipse (0.5 and 1) node[above](X1){} node[below](X2){};
        \node at (-1,1.3){$\mathcal{X}$};
        \draw (2, 0) ellipse (1 and 2);
        \node at (2, 2.3){$\mathcal{Y}$};
        \draw (2, 0.7) circle (0.5) node[left](Y1){};
        \draw (2, -0.7) circle (0.5) node[left](Y2){};
        \draw[->, bend left=20] (X1) to node[midway, above] {$f_{\cdot, 0}$} (Y1);
        \draw[->, bend left=-10] (X2) to node[midway, below] {$f_{\cdot, 1}$} (Y2);
    \end{tikzpicture}
    \subcaption{$\mu = 1$}
\end{subfigure}
\caption{$4$-to-$2$ Dual-mode trapdoor functions ($4$-to-$2$ dTFs): In mode $\mu = 0$, the functions $f_{b_1, b_2}, f_{b'_1, b'_2}$ have the same image if and only if $b_1 = b'_1$, and otherwise they have disjoint images. In mode $\mu =1$, the functions $f_{b_1, b_2}, f_{b'_1, b'_2}$ have the same image if and only $b_2 = b'_2$ and otherwise they have disjoint images. The two modes are computationally indistinguishable.\label{figure:4-2-dTF}}
\end{figure}

\begin{definition}[$\epsilon$-weak $4$-to-$2$ dual-mode trapdoor functions]
Let $\lambda$ be the security parameter, and let $\mathcal{X}$, $\mathcal{Y}$ be finite sets. Let $\mathcal{K}$ be a finite set of keys. A family of functions
\begin{align*}
    \mathcal{F} = \{f_{k, b_1, b_2} : \mathcal{X} \rightarrow \mathcal{Y}\}_{k \in \mathcal{K}, b_1, b_2 \in \{0,1\}}
\end{align*}
is a \textbf{$\epsilon$-weak $4$-to-$2$ dual-mode trapdoor function family} if there exists a family of probability distributions $\mathcal{D} = \{D_{k, b_1, b_2}\}_{k \in \mathcal{K}, b_1, b_2 \in\{0,1\}}$ over $\mathcal{X}$ such that the following conditions hold.
\begin{enumerate}
    \item{\textbf{Efficient key generation.}} There is an efficient probabilistic algorithm $\mathsf{Gen}_\mathcal{F}$ which takes as input the security parameter $1^\lambda$ and a ``mode'' bit $\mu \in\{0,1\}$, and generates a key $k(\mu) \in \mathcal{K}$ together with a trapdoor $t_k$,
    \begin{align*}
        (k(\mu), t_k) \leftarrow \mathsf{Gen}_{\mathcal{F}}(1^\secp, \mu).
    \end{align*}
    \item{\textbf{Efficient function evaluation.}} There is an efficient deterministic algorithm that takes as input $k\in\mathcal{X}$, $b_1, b_2\in\{0,1\}$, $x\in\mathcal{X}$ and computes $f_{k, b_1, b_2}(x)$.
    \item{\textbf{Efficient state preparation.}} There is a quantum polynomial-time algorithm that prepares for every $k \in \mathcal{K}$, $b_1, b_2\in\{0,1\}$, a state that is negligibly close in trace distance to the state $\sum_{x \in \mathcal{X}} \sqrt{D_{k,b_1,b_2}(x)} \ket{x}$.
    \item{\textbf{Efficient partial inversion and phase computation with the trapdoor.}} There exists an efficient deterministic algorithm that takes as input $t_k, y$ and outputs the set 
    \begin{align*}
        \left\{(b_1, b_2) \in \{0,1\}^2 : \exists x \in \mathcal{X} \text{ s.t. } f_{k, b_1, b_2}(x) = y\right\}.
    \end{align*}
    Additionally, there exists another efficient deterministic algorithm that takes as input $y\in\mathcal{Y}$, $d\in\bar{\mathcal{X}}$ and $b_1, b_2 \in \{0,1\}$, along with trapdoor $t_k$, and computes (up to a common normalization factor) $\alpha_{k,y, d}(b_1, b_2)$, defined as
    \begin{align*}
        \alpha_{k, y, d}(b_1, b_2) := \sum_{x \in \mathcal{X}} (-1)^{d \cdot x} \sqrt{D_{k, b_1, b_2}(x)}\;.
    \end{align*}
    When $k$, $y$ and $d$ are clear from the context, we will sometimes abbreviate $\alpha_{k, y, d}(b_1, b_2)$ as $\alpha_{b_1, b_2}$.
    \item{\textbf{Dual-mode.}} Intuitively, we want the functions $f_{k, b_1, b_2}$ and $f_{k, b'_1, b'_2}$ to have identical or disjoint images, depending on whether $b_{\mu+1} = b'_{\mu+1}$, in mode $\mu$. More precisely, for all but a negligible fraction of keys $(k(\mu), \cdot) \leftarrow \mathsf{Gen}_{\mathcal{F}}(1^\lambda, \mu)$,
    \begin{enumerate}
        \item Case 1 ($b_{\mu+1} \neq b'_{\mu+1}$): For $x\leftarrow D_{k, b_1, b_2}$, $x'\leftarrow D_{k, b'_1, b'_2}$, we have that $f_{k, b_1, b_2}(x) \neq f_{k, b'_1, b'_2}(x')$ with probability $1$.
        \item Case 2 ($b_{\mu+1} = b'_{\mu+1}$): For $\tilde{b}_1, \tilde{b}_2 \leftarrow \{0,1\}$, $x \leftarrow D_{k, \tilde{b}_1, \tilde{b}_2}$, $y = f_{k, \tilde{b}_1, \tilde{b}_2}(x)$, for all $d \in \bar{\mathcal{X}}$, with probability at least $1-\epsilon$, it must be the case that there exists some $s \in \{0,1\}$ such that $\alpha_{b_1, b_2} = (-1)^s \cdot \alpha_{b'_1, b'_2}$.
    \end{enumerate}
    \item{\textbf{Mode Indistinguishability.}} Given a key $k(\mu)$ generated with $\mathsf{Gen}(1^\lambda, \mu)$, no efficient algorithm can guess $\mu$ with greater than negligible advantage.
\end{enumerate}
\end{definition}

\begin{lemma}\label{lemma:dTF-to-4-2-dTF}
Suppose there is an $\epsilon$-weak dual-mode trapdoor function family $\mathcal{F}$. Then, there exists an $\epsilon$-weak $4$-to-$2$ dual-mode trapdoor function family $\mathcal{G}$.
\end{lemma}

\begin{proof}
Suppose $\mathcal{F} = \{f_{k, b} : \mathcal{X} \rightarrow \mathcal{Y}\}_{k \in \mathcal{K}_{\mathcal{F}}, b \in \{0,1\}}$ is a $\epsilon$-weak dual-mode trapdoor function family. Define
\begin{align*}
    \mathcal{G} = \{g_{k, b_1, b_2}: \mathcal{X} \times \mathcal{X} \rightarrow \mathcal{Y} \times \mathcal{Y}\}_{k \in \mathcal{K}_{\mathcal{G}}, b_1, b_2 \in \{0,1\}}
\end{align*}
for some $\mathcal{K}_\mathcal{G} \subseteq \mathcal{K}_{\mathcal{F}} \times \mathcal{K}_{\mathcal{F}}$.
such that for all $k = (k_1, k_2) \in \mathcal{K}_{\mathcal{G}}$, $b_1, b_2 \in \{0,1\}$, and $x_1, x_2 \in \mathcal{X}$, define
\begin{align*}
    g_{k, b_1, b_2}(x_1, x_2) := \Bigl( f_{k_1, b_1}(x_1),  f_{k_2, b_2}(x_2) \Bigr).
\end{align*}
For convenience, we will sometimes write $f_1 = f_{k_1, b_1}$ and $f_2 = f_{k_2, b_2}$. We now show that $\mathcal{G}$ is a $4$-to-$2$ dual-mode trapdoor function family.

\paragraph{Efficient key generation and function evaluation.} We define $\mathsf{Gen}_{\mathcal{G}}(1^\lambda, \mu)$ to work as follows. We sample the key-trapdoor pairs
\begin{align*}
    (k_1, t_{k_1}) &\leftarrow \mathsf{Gen}_{\mathcal{F}}(1^\lambda, \mu), \text{ and}\\
    (k_2, t_{k_2}) &\leftarrow \mathsf{Gen}_{\mathcal{F}}(1^\lambda, 1-\mu).
\end{align*}
We output the key $k = (k_1, k_2)$ and the trapdoor $t_k = (t_{k_1}, t_{k_2})$. Since $\mathsf{Gen}_{\mathcal{F}}$ is efficient, so is $\mathsf{Gen}_{\mathcal{G}}$. Efficient function evaluation follows directly from the definition.

\paragraph{Efficient state preparation.} If $\mathcal{D} = \{D_{k, b}\}_{k\in\mathcal{K}_{\mathcal{F}}, b \in \{0,1\}}$ is the family of probability distributions over $\mathcal{X}$ associated with $\mathcal{F}$.
Let $\mathcal{D}' = \{D'_{k,b_1, b_2}\}_{k = (k_1, k_2) \in \mathcal{K}_{\mathcal{G}}, b_1, b_2\in\{0,1\}}$ be the family of distributions associated with $\mathcal{G}$, defined as $D'_{k, b_1, b_2}(x_1, x_2) := D_{k_1, b_1}(x_1) \cdot D_{k_2, b_2}(x_2)$. Then we can efficiently prepare the state as a tensor product:
\begin{align*}
    \sum_{x_1, x_2\in \mathcal{X}} & \sqrt{D'_{k, b_1, b_2}(x_1, x_2)}\ket{x_1, x_2} \\ & = \sum_{x_1 \in \mathcal{X}} \sqrt{D_{k_1, b_2}(x_1)} \ket{x_1} \otimes \sum_{x_2 \in \mathcal{X}} \sqrt{D_{k_2, b_2}(x_2)} \ket{x_2}.
\end{align*}

\paragraph{Efficient partial inversion with the trapdoor.} Given as input $y = (y_1, y_2)$ and the trapdoor $t_k = (t_{k_1}, t_{k_2})$, we can partially invert $f_{k_1, \cdot}$ and $f_{k_2, \cdot}$ on $y_1$ and $y_2$ respectively to get $B_1 = \{b_1 : \exists f_{k_1, b_1}(x) = y_1\}$ and $B_2 = \{b_2 : \exists f_{k_2, b_2}(x) = y_2\}$. Output $B_1 \times B_2$.

\paragraph{Efficient phase computation with the trapdoor.}
Fix $y = (y_1, y_2) \in \mathcal{Y}^2$ and $d = (d_1, d_2) \in \bar{\mathcal{X}}^2$. Then,
\begin{align}
    \alpha_{k, y, d}(b_1, b_2) &= \sum_{x_1, x_2 \in \mathcal{X}} (-1)^{d_1 \cdot x_1 + d_2 \cdot x_2} \cdot \sqrt{D'_{k, b_1, b_2}(x_1, x_2)}\nonumber\\
    &= \sum_{x_1, x_2 \in \mathcal{X}} (-1)^{d_1 \cdot x_1 + d_2 \cdot x_2} \cdot \sqrt{D_{k_1, b_1}(x_1) \cdot D_{k_2, b_2}(x_2)}\nonumber\\
    &= \left(\sum_{x_1 \in \mathcal{X}}  (-1)^{d_1 \cdot x_1} \sqrt{D_{k_1, b_1}(x_1)}\right) \cdot \left(\sum_{x_2 \in \mathcal{X}} (-1)^{d_2 \cdot x_2} \sqrt{D_{k_2, b_2}(x_2)}\right)\nonumber\\
    &= \alpha_{k_1, y_1, d_1}(b_1) \cdot \alpha_{k_2, y_2, d_2}(b_2)\;.\label{equation:phase-computation-4-2}
\end{align}
Given the trapdoor, since we can effiently compute $\alpha_{k_1, y_1, d_1}(b_1)$ and $\alpha_{k_2, y_2, d_2}(b_2)$ individually, we can also compute $\alpha_{k, y, d}(b_1, b_2)$.

\paragraph{Dual mode.} We will prove the dual-mode property for the case where $\mu = 0$; the case where $\mu = 1$ is symmetric. When $\mu = 0$, the keys $k_1(0)$ and $k_2(1)$ were generated in modes $0$ and $1$ respectively.

If $b_1 \neq b'_1$, since $k_1(0)$ is in mode $0$, we know that the images of $f_{k_1, b_1}$ and $f_{k_1, b'_1}$ are disjoint, and so are the images of $g_{k, b_1, b_2}$ and $g_{k, b'_1, b'_2}$. 

If $b_1 = b'_1$, consider sampling $\tilde{b}_1, \tilde{b}_2 \leftarrow \{0,1\}$, $x_1 \leftarrow D_{k_1, \tilde{b}_1}$, $x_2 \leftarrow D_{k_2, \tilde{b}_2}$, and setting $y_1 = f_{k_1, \tilde{b}_1}(x_1)$ and $y_2 = f_{k_2, \tilde{b}_2}(x_2)$. Since $k_2(1)$ is in mode $1$, we know with probability at least $1-\epsilon$, for any $d_2\in \bar{\mathcal{X}}$ there exists some $s\in\{0,1\}$ such that $\alpha_{k_2, y_2, d_2}(b_2) = (-1)^s \cdot \alpha_{k_2, y_2, d_2}(b'_2)$. By Equation~\eqref{equation:phase-computation-4-2}, we get that with probability at least $1-\epsilon$ for any $d_1, d_2\in \bar{\mathcal{X}}$,
\begin{align*}
    \alpha_{k, y, d}(b_1, b_2) &= \alpha_{k_1, y_1, d_1}(b_1) \cdot \alpha_{k_2, y_2, d_2}(b_2) = \alpha_{k_1, y_1, d_1}(b'_1) \cdot \alpha_{k_2, y_2, d_2}(b_2)\\
    &= \alpha_{k_1, y_1, d_1}(b'_1) \cdot (-1)^s \cdot \alpha_{k_2, y_2, d_2}(b'_2) =  (-1)^s \cdot \alpha_{k, y, d}(b'_1, b'_2).
\end{align*}

\paragraph{Mode-Indistinguishability.} Consider the following hybrid experiments:
\begin{enumerate}
    \item Hybrid 1: Generate $k = (k_1, k_2)$ by running $k_1, t_1 \leftarrow \mathsf{Gen}_{\mathcal{F}}(1^\lambda, \mu)$ and $k_2, t_2 \leftarrow \mathsf{Gen}_{\mathcal{F}}(1^\lambda, 1 - \mu)$, and output $k = (k_1, k_2)$.
    \item Hybrid 2: Generate $k = (k_1, k_2)$ by running $k_1, t_1 \leftarrow \mathsf{Gen}_{\mathcal{F}}(1^\lambda, \mu)$ and $k_2, t_2 \leftarrow \mathsf{Gen}_{\mathcal{F}}(1^\lambda, \mu)$, and output $k = (k_1, k_2)$.
    \item Hybrid 3: Generate $k = (k_1, k_2)$ by running $k_1, t_1 \leftarrow \mathsf{Gen}_{\mathcal{F}}(1^\lambda, 1 - \mu)$ and $k_2, t_2 \leftarrow \mathsf{Gen}_{\mathcal{F}}(1^\lambda, \mu)$, and output $k = (k_1, k_2)$.
\end{enumerate}
Note that Hybrid 1 corresponds to the setting where $k$ is generated as $k, t \leftarrow \mathsf{Gen}_{\mathcal{G}}(1^\lambda, \mu)$, and Hybrid 3 corresponds to the setting where $k$ is generated as $k, t \leftarrow \mathsf{Gen}_{\mathcal{G}}(1^\lambda, 1 - \mu)$.
By the mode-indistinguishability of $\mathcal{F}$, Hybrids 1 and 2 are computationally indistinguishable, and similarly Hybrids 2 and 3 are computationally indistinguishable. Therefore, Hybrids 1 and 3 are computationally indistinguishable. This completes the proof.
\end{proof}

\subsection{Amplification Lemma for Dual-Mode Trapdoor Functions}\label{section:dTF-amplification}
In this section, we give an amplification lemma that transforms $\epsilon$-weak injective and invertible dTFs (see Definition~\ref{definition:injective-and-invertible}) to $\epsilon^\ell$-weak dTFs. This transformation is based on the XOR lemma of \cite{sahai2003complete}.
\ifnum\llncs=1
See the full version of the paper for the proof.
\fi

\begin{lemma}[Amplification Lemma for dTFs] \label{lemma:dTF-amplification}
Let $\lambda$ be the security parameter. Suppose there exists an injective and invertible $\epsilon$-weak dual-mode trapdoor function family for some $\epsilon > 0$. Then, there exists a $\epsilon^\ell$-weak dual-mode trapdoor function family for any $\ell \le \poly(\lambda)$.
\end{lemma}

\ifnum\llncs=0
\begin{proof}
Let $\mathcal{F} = \{f_{k, b} : \mathcal{X} \rightarrow \mathcal{Y}\}_{k \in \mathcal{K}, b \in \{0,1\}}$ be an injective and invertible, $\epsilon$-weak dual-mode trapdoor function family. Define the function family $\mathcal{G}$ as
\begin{gather*}
    \mathcal{G} = \{g_{k, b} : \mathcal{X}^\ell \times \{0,1\}^{\ell -1} \rightarrow \mathcal{Y}^\ell\}_{k \in \mathcal{K}, b \in \{0,1\}},\\
    g_{k, b}(x_1, \ldots, x_\ell, r_1, \ldots, r_{\ell-1}) := \left(y_1, \ldots, y_\ell\right),
\end{gather*}
where $r_\ell = b \oplus r_1 \oplus \cdots \oplus r_{\ell}$ and $y_i = f_{k, r_i}(x_i)$ for $i \in [\ell]$. We will show that $\mathcal{G}$ is a $\epsilon^\ell$-weak dual-mode trapdoor function family. If $\mathcal{D} = \{D_{k, b}\}_{k, b}$ is the family of probability distributions associated with $\mathcal{F}$, we define a family of probability distributions $\mathcal{D}' = \{D'_{k, b}\}_{k, b}$ associated with $\mathcal{G}$ as follows:
\begin{align*}
    D'_{k, b}(x_1, \ldots, x_\ell, r_1, \ldots, r_{\ell-1}) := \frac{1}{2^{\ell-1}} \cdot \prod_{i=1}^\ell D_{k,r_i}(x_i),
\end{align*}
where $r_\ell = b \oplus r_1 \oplus \cdots \oplus r_{\ell}$.

\paragraph{Key generation, efficient state preparation, and mode indistinguishability.} The key generation algorithm for $\mathcal{G}$ is the same as that of $\mathcal{F}$, so the efficient key generation and mode indistinguishability directly for $\mathcal{G}$ follow from the same properties of $\mathcal{F}$. The efficient state preparation property for $\mathcal{G}$ also follows directly from that of $\mathcal{F}$.

\paragraph{Efficient function evaluation.} Given the key, the function can also be evaluated efficiently: to evaluate $g_{k, b}$ on input $x_1, \ldots, x_\ell, r_1, \ldots, r_{\ell-1}$, we first compute $r_\ell := b \oplus r_1 \oplus \cdots \oplus r_\ell$. Since $f_{k, \cdot}$ can be efficiently evaluated, we compute $y_i = f_{k, r_i}(x_i)$ for $i \in [\ell]$ and output $(y_1, \ldots, y_\ell)$.

\paragraph{Efficient partial inversion with a trapdoor.} There is also a straightforward efficient algorithm that performs partial inversion: given the trapdoor $t_k$ and $y = (y_i)_{i\in [\ell]}$, for each $y_i$, run the partial inversion algorithm for $\mathcal{F}$ to the set $R_i := \{r \in \{0,1\} : \exists x \text{ s.t. } f_{k, r}(x) = y_i\}$. On a valid input, every $R_i$ must be non-empty; otherwise output $\emptyset$. If $|R_i| = 1$ for every $i$ (so we can write $R_i = \{\hat{r}_i\}$), then there is only on possible value of $b$, so we output $\{\hat{b}\}$, where $\hat{b} := \hat{r}_1 \oplus \cdots \oplus \hat{r}_\ell$. If for any $i$, $R_i = \{0,1\}$, this means that $y$ has preimages under both $g_{k, 0}$ and $g_{k, 1}$, so we output $\{0,1\}$.

For the rest of the proof, we focus on efficient phase computation and the dual-mode property. The proof is by induction, and it will be useful to first set up some notation. For a fixed $y = (y_1, \ldots, y_\ell)$ in the image of $g_{k, b}$, and for a fixed $(d, s) = (d_1, \ldots d_\ell, s_1, \ldots, s_{\ell-1}) \in \bar{\mathcal{X}}^\ell \times \{0,1\}^{\ell-1}$, define $S_{i, b}$ and $\beta_{i, b}$ for each $i \in [\ell]$ as follows:
\begin{align*}
    S_{i, b} &:= \left\{(x_1, \ldots, x_{i}, r_1, \ldots, r_{i-1}): \exists r_i \in \{0,1\} \text{ s.t. } r_1 \oplus \cdots \oplus r_{i} = b, f_{k, r_1}(x_1) = y_1, \ldots, f_{k, r_{i}}(x_{i}) = y_{i}\right\},
\end{align*}
and
\begin{align*}
    \beta_{i, b} &:= \frac{1}{2^{i-1}} \cdot \sum_{\substack{x_1, \ldots, x_{i},\\r_1, \ldots, r_{i-1} \in S_{i, b}}}  \left( \prod_{j=1}^{i} (-1)^{d_j \cdot x_j} \cdot \sqrt{D_{k, r_j}(x_j)} \right) \cdot \left(\prod_{j=1}^{i-1} (-1)^{s_j \cdot r_j}\right)\\
    &= \frac{1}{2^{i-1}} \cdot \sum_{\substack{x_1, \ldots, x_{i},\\r_1, \ldots, r_{i-1} \in S_{i, b}}} (-1)^{\sum_{j=1}^{i} d_j \cdot x_j + \sum_{j=1}^{i-1} s_j \cdot r_j} \cdot \left( \prod_{j=1}^{i} \sqrt{D_{k, r_j}(x_j)} \right)
\end{align*}
Note that $\alpha_{k, y, (d,s)}(b) = \beta_{\ell, b}$, and $S_{\ell, b}$ is the set of all preimages of $y$ under $g_{k, b}$. For $i \in [\ell]$, $r \in \{0,1\}$, let $\mathbbm{1}[X_{i, r}]$ be the indicator for the event that there exists some $x_{i, r}\in\mathcal{X}$ such that $f_{k, r}(x_{i, r}) = y_i$. Since $\mathcal{F}$ is an injective and invertible dTF family, such an inverse $x_{i, r}$ must be unique if it exists. Writing $\beta_{i, r}$ in terms of $\beta_{i-1, 0}$ and $\beta_{i-1, 1}$, we get
\begin{align}
    \beta_{i, b} &= \mathbbm{1}[X_{i, 0}] \cdot \frac{1}{2} \cdot (-1)^{d_i \cdot x_{i,0}} \cdot \sqrt{D_{k, 0}(x_{i,0})} \cdot \sum_{\substack{x_1, \ldots, x_{i-1},\\r_1, \ldots, r_{i-2} \in S_{i-1, b}}} (-1)^{\sum_{j=1}^{i-1} d_j \cdot x_j + \sum_{j=1}^{i-2} s_j \cdot r_j} \cdot \prod_{j=1}^{i-1}\sqrt{D_{k, r_j}(x_j)} \nonumber \\
    &\, + \mathbbm{1}[X_{i, 1}] \cdot \frac{1}{2} \cdot (-1)^{d_i \cdot x_{i, 1} + s_{i-1}} \cdot \sqrt{D_{k, 1}(x_{i, 1})} \cdot \sum_{\substack{x_1, \ldots, x_{i}\\r_1, \ldots, r_{i-2}\\\in S_{i-1, 1-b}}} (-1)^{\sum_{j=1}^{i-1} d_j \cdot x_j  + \sum_{j=1}^{i-2} s_j \cdot r_j} \cdot \prod_{j=1}^{i-1}\sqrt{D_{k, r_j}(x_j)} \nonumber \\
    &= \frac{1}{2} \sqrt{D_{k, 0}(x_{i,0})} \cdot \left(\mathbbm{1}[X_{i, 0}] \cdot (-1)^{d_i \cdot x_{i,0}} \cdot \beta_{i-1, b}  + \mathbbm{1}[X_{i, 1}] \cdot (-1)^{d_i \cdot x_{i, 1} + s_{i-1}} \cdot \beta_{i-1, 1-b} \right)\;,\label{equation:induction}
\end{align}
where the last inequality follows since $D_{k, 0}(x_{i, 0}) = D_{k, 1}(x_{i, 1})$, since $\mathcal{F}$ is an injective and invertible dTF.
\paragraph{Efficient Phase Computation with Trapdoor.} We observe that in the base case, $$S_{1, b} = \{x_1 : f_{k, b}(x_1) = y_1\},$$ and since $\mathcal{F}$ is an injective and invertible dTF family, $\beta_{1, b}$ is efficiently computable. For the same reason, the phases $(-1)^{d_i \cdot x_{i,0}}$ and $(-1)^{d_i \cdot x_{i, 1} + s_{i-1}}$ are efficiently computable for all $i \in [\ell]$. Then, by Equation~\eqref{equation:induction}, we can sequentially compute $\beta_{i, r}$ for $i = 1, \ldots, \ell$, and $r\in\{0,1\}$, to get finally $\alpha_b = \alpha_{k, y, (d,s)}(b) = \beta_{\ell, b}$.

\paragraph{Dual-mode.} In the disjoint mode ($\mu = 0$), since the images of $f_{k, 0}$ and $f_{k, 1}$ are disjoint, for each $i \in [\ell]$, there exist unique $r_i, x_i$ such that $f_{k, r_i}(x_i) = y_i$. This implies that there is a unique $b = r_1 \oplus \cdots \oplus r_\ell$ such that $y$ is in the image of $g_{k, b}$. In the lossy mode ($\mu = 1$), for $i\in[\ell]$ letting $z_i \in \{0,1\}$ such that $(-1)^{z_i} = (-1)^{d_i \cdot x_{i, 1} + s_{i-1} - d_i \cdot x_{i,0}}$, Equation~\eqref{equation:induction} implies
\begin{align*}
    |\beta_{i, b}| = \frac{1}{2} \cdot \sqrt{D_{k, 0}(x_{i, 0})} \cdot \Bigl|\mathbbm{1}[X_{i, 0}] \cdot \beta_{i-1, b}  + \mathbbm{1}[X_{i, 1}] \cdot (-1)^{z_i} \cdot \beta_{i-1, 1-b}\Bigr|\;.
\end{align*}
If $\mathbbm{1}[X_{i, 0}] = 1$ and $\mathbbm{1}[X_{i, 1}] = 0$, we have that $|\beta_{i, b}| = \frac{1}{2} \cdot \sqrt{D_{k, 0}(x_{i, 0})} \cdot |\beta_{i-1, b}|$. Symmetrically, if on the other hand $\mathbbm{1}[X_{i, 0}] = 0$ and $\mathbbm{1}[X_{i, 1}] = 1$, we have that $|\beta_{i, b}| = |\beta_{i-1, 1-b}|$. In the case where $\mathbbm{1}[X_{i, 0}] = \mathbbm{1}[X_{i, 1}] = 1$, 
\begin{align*}
    |\beta_{i, b}| = \frac{1}{2} \cdot \sqrt{D_{k, 0}(x_{i, 0})} \cdot \Bigl|\beta_{i-1, b} + (-1)^{z_i} \beta_{i-1, 1-b} \Bigr| = |\beta_{i, 1-b}| \;.
\end{align*}
Therefore, in all cases, if $|\beta_{i-1, 0}| = |\beta_{i-1, 1}|$, we have that $|\beta_{i,0}| = |\beta_{i, 1}|$. 

All that remains is to prove the base case: we want to show that there exists some $i^* \in [\ell]$ for which $|\beta_{i^*,0}| = |\beta_{i^*, 1}|$. By the lossy-mode property of $\epsilon$-weak dTF $\mathcal{F}$, with probability\footnote{the probability is taken over the choice of $y$, which is sampled by first sampling $x_1, \ldots, x_\ell \leftarrow \mathcal{X}$ and $r_1, \ldots, r_{\ell-1} \leftarrow \{0,1\}$ and setting $y = g_{k, b}(x_1, \ldots, x_\ell, r_1, \ldots, r_{\ell-1})$. Let $i^*$ be the smallest such $i$.} at least $1 - \epsilon^\ell$, there exists some $i \in [\ell]$ such that $\mathbbm{1}[X_{i,0}] = \mathbbm{1}[X_{i,1}] = 1$. Let $i^*$ be the smallest such $i$. For all $j < i^*$, exactly one of $|\beta_{j, 0}|$ and $|\beta_{j, 1}|$ is nonzero, the other being $0$. Further, $|\beta_{i^*, 0}| = |\beta_{i^*, 1}| = \frac{1}{2} \cdot \sqrt{D_{k, 0}(x_{i^*,0})}$, 
satisfying the base case for our induction. Therefore in the lossy mode, with probability at least $1 - \epsilon^\ell$, $|\beta_{\ell, 0}| = |\beta_{\ell, 1}|$. Since the $\beta_{i, r}$'s are real-valued, this means that there exists some $s\in \{0,1\}$ such that $\beta_{\ell, 0} = (-1)^s \cdot \beta_{\ell, 1}$. This completes the proof.
\end{proof}
\fi

\section{Remote State Preparation of the DSS  Gadget}\label{section:dss-rsp}
We construct our QFHE scheme by dequantizing the client in the scheme of Dulek, Schaffner and Speelman~\cite{dulek2016quantum}, which we will henceforth refer to as DSS. The DSS QFHE scheme builds on the Pauli one-time pad-based template of Broadbent and Jeffery~\cite{broadbent2015quantum}. Concretely, a quantum state $\ket{\psi}$ is encrypted by applying a random Pauli one-time pad to get $|\tilde{\psi}\rangle := \X^x \Z^z \ket{\psi}$. The keys $x, z$ of the one-time pad are encrypted under a classical homomorphic encryption scheme to get ciphertexts $\tilde{x}, \tilde{z}$. The one-time padded state $|\tilde{\psi}\rangle$ along with these encryptions $\tilde{x}, \tilde{z}$ form the encryption of the state $\ket{\psi}$.

Evaluating any Clifford gate is easy, since Clifford gates nicely ``commute'' through the Pauli one-time pad; we only need to update the one-time pads homomorphically according to standard rules of Clifford-Pauli commutation. On the other hand when the non-Clifford $\T$ gate is applied to a Pauli one-time padded state, we end up with an extra non-Pauli error in the form of a $\P$ gate, since
\begin{align*}
    \T \X^x \Z^z \ket{\psi} = \P^x \X^x \Z^z \T \ket{\psi}.
\end{align*}
The natural solution would be to apply a $\P^{\dagger}$ gate to this state if $x = 1$, and otherwise do nothing. However it is not clear how to do this, because the evaluator only has access to an encryption $\tilde{x}$ of the one-time pad key $x$. 

Dulek et al.~\cite{dulek2016quantum} solve this issue by ``obfuscating'' the decryption algorithm (hard-coded with the secret key $sk$) in the form of a quantum teleportation gadget $\Gamma(sk)$, henceforth called the DSS gadget or the DSS teleportation gadget. The client prepares this (quantum) gadget during key generation, and sends it to the evaluator as part of the evaluation key. When the qubit $\P^x \X^x \Z^z \T \ket{\psi}$ is teleported through the gadget (in a way that depends on $\tilde{x}$), the gadget implicitly decrypts $\tilde{x}$ and applies the $(\P^\dagger)^x$ correction to the teleported qubit. More specifically, we can think of running the DSS gadget as a blackbox that takes as input a state $\rho$ and an encryption $\HE.\Enc_{pk}(x)$ for some $x \in \{0,1\}$ and outputs the state
\begin{align*}
    \X^{x'} \Z^{z'} (\P^\dagger)^x \rho (\P)^x \Z^{z'} \X^{x'},
\end{align*}
for some one-time pad keys $x', z' \in \{0,1\}$, along with encryptions $\HE.\Enc_{pk'}(x')$ and $\HE.\Enc_{pk'}(z')$ under an independent public key $pk'$. This functionality is visualized in Figure~\ref{fig:DSS}, up to a renaming of the variables $x', z'$.

In this section, we describe a protocol by which the evaluator is able to prepare the DSS gadget $\Gamma(sk)$ on its own, with the help of some {\em classical} information generated by the client during the key generation procedure. In other words, we give a remote state preparation protocol for the DSS gadget. As a result, the client is entirely classical. This protocol is described in Section~\ref{section:dss-gadget-rsp-protocol}. The complete description of the structure of the DSS gadget is quite involved, so we give only an abstract description of its structure in Section~\ref{section:abstract-description-DSS}, which will suffice for our purposes. We defer the details of the structure of the DSS gadget to %
\ifnum\llncs=0
Appendix~\ref{section:DSSgadget}.
\else
Appendix~A of the full version.
\fi

At the heart of our remote state preparation protocol is a simpler protocol that takes as input a three-qubit state, say on registers labeled $0, 1, 2$, along with a ``$4$-to-$2$'' dual-mode trapdoor function key $k(\mu)$ generated in mode $\mu$. The output is the (one-time padded) state $\ket{\Phi^+}_{2, \mu} \otimes \ket{+}_{1-\mu}$, where the qubit in register $\mu$ forms a Bell pair with qubit $2$, and these two registers are in tensor product with the qubit in register $(1-\mu)$. In other words, this protocol allows the remote state preparation of a Bell pair between two of three qubits, while hiding the identity of the qubits that form the Bell pair. This protocol for remote state preparation of a ``hidden Bell pair'' is described in Section~\ref{section:RSPhiddenBellPair}.

\subsection{An abstract description of the DSS gadget}\label{section:abstract-description-DSS}
The DSS gadget is composed of a set of entangled qubits to form a perfect matching of Bell pairs. To facilitate teleportation, the gadget also defines a set of Bell-basis measurements for the evaluator to perform. 

The complete description of this matching and measurement instructions is quite involved, so we defer our the exposition to %
\ifnum\llncs=0
Appendix~\ref{section:DSSgadget}. %
\else
Appendix~A of the full version. %
\fi
For the purpose of this section, we only need an abstract specification of the structure of the DSS gadget, which we will now describe.

The DSS gadget is defined by a tuple $(Q, Q_1, Q_2, P, \pi, \nu)$, where each element of the tuple is described below. Let $n = |sk|$ be the length of the secret key, and let $\mathcal{C}$ be the space of ciphertexts for a single bit message for the classical homomorphic encryption scheme $\HE$.
\begin{itemize}
    \item $Q$ is the set of qubit registers in the DSS gadget.
    \item $(Q_1, Q_2)$ is a partition of $Q$, across which the entangled Bell pairs form a bipartite graph.
    \item $P \subseteq Q$ is the set of registers with an inverse phase gate $\P^\dagger$ applied.
    \item $\pi :  \{0,1\}^n \times Q_1 \rightarrow Q_2$ is an injective function that defines the matching of qubits that form the maximally entangled Bell pairs of the gadget. If $sk$ is the secret key, then for every $q\in Q_1$, the qubits in registers $q, \pi(sk, q)$ form a Bell pair. We write $\pi^{sk}$ as shorthand for $\pi(sk, \cdot)$.\footnote{Although $\pi$ is public and known to the evaluator, $\pi^{sk}$ is still private and known only to the client.}
    \item $\nu: Q_2 \times \mathcal{C} \rightarrow Q_1 \cup \{\bot\}$ defines the measurements the evaluator must perform to teleport the input qubit through the gadget. These measurements depend on the ciphertext $\tilde{x}$, which the evaluator has access to. For every $q \in Q_2$, the evaluator performs a Bell basis measurement on the qubits $q, \nu(\tilde{x}, q)$ (in a specified order) and records the measurements to update the one-time pad keys.
\end{itemize}

All of $Q, Q_1, Q_2, P, \pi, \nu$ are public and known to the evaluator. Together with the secret key $sk$, they determine the structure of the DSS gadget, that is, they determine which qubit pairs are entangled to form Bell pairs. The function $\pi^{sk} := \pi(sk, \cdot)$ defines this bijection from $Q_1$ to $Q_2$. Concretely, the gadget $\Gamma(sk)$ consists of the state
\begin{align*}
    \gamma_{x, z}(sk) := \bigotimes_{q \in Q_1} \X^{x_{q}} \Z^{z_{q}} (\P^\dagger)^{p_{q}} \ket{\Phi^+}\bra{\Phi^+}_{q, \pi^{sk}(q)} (\P)^{p_{q}} \Z^{z_{q}} \X^{x_{q}},
\end{align*}
where $x, z \in \{0,1\}^{|Q_1|}$ are one-time pad keys and $p \in \{0,1\}^{|Q_1|}$ is defined as $p_q = 1$ if $q \in P$ and $0$ otherwise. The gadget $\Gamma(sk)$ also includes classical encryptions $\tilde{x}, \tilde{z}$ of the one-time pad keys under an independent public key $pk'$. We write $\Gamma_{pk'}(sk) = (\gamma_{x, z}(sk), \tilde{x}, \tilde{z})$ to denote this.

The function $\pi$ has the additional property that for any $q \in Q_1$, there exists some $i(q)$ such that $sk_{i(q)}$ determines the value of $\pi(\cdot, q)$. That is, $\pi(\cdot, q)$ can take one of two distinct values, and depends on only the $i(q)$th bit of $sk$ (as opposed to some global function of $sk$). We denote these two possibilities as $\pi_0(q)$ and $\pi_1(q)$, the value of $\pi(sk, q)$ when $sk_{i(q)} = 0$ and $sk_{i(q)} = 1$ respectively.

\subsection{Remotely Preparing a Hidden Bell Pair}\label{section:RSPhiddenBellPair}
In this subsection, we give a protocol to remotely prepare a hidden Bell pair. At a high level, given a key $k(\mu)$ of a $4$-to-$2$ dual-mode trapdoor function (which can be instantiated given just a dual-mode trapdoor function), and three registers in the state $\ket{+}_0 \otimes \ket{+}_1 \otimes \ket{+}_2$, the protocol produces a one-time padded version of the following state
\begin{align*}
    \ket{\Phi^+}_{2, \mu} \otimes \ket{+}_{1 - \mu}.
\end{align*}
To use this protocol in the remote state preparation of the DSS gadget, we also need to consider input states which are one-time padded with $\Z$ gates and the states in registers $0$ and $1$ could be potentially entangled with other registers.

\begin{lemma}\label{lemma:RSP-hidden-bell-pair}
Let $\mathcal{F}$ be a $\epsilon$-weak $4$-to-$2$ dual-mode trapdoor function family, and let $\HE$ be a fully homomorphic encryption scheme. For a bit $\mu \in \{0, 1\}$, let $k, t \leftarrow \KeyGen_{\mathcal{F}}(1^\lambda, \mu)$ be a key-trapdoor pair, let $(pk, evk, sk) \leftarrow \HE.\KeyGen(1^\lambda)$ be a key set, and let $\tilde{t} \leftarrow \HE.\Enc_{pk}(t)$ be an encryption of the trapdoor $t$ under $pk$. Consider a BQP machine with access to the state\footnote{The single-qubit gates $\Z^{z_0}, \Z^{z_1}$ are applied to the first qubits of the $0,1$ registers.}
\begin{align}\label{equation:hidden-bell-pair-input}
    (\Z^{z_0} \otimes \Z^{z_1} \otimes \I) \Bigl(\ket{\psi_0}_0 \otimes \ket{\psi_1}_1 \otimes \ket{+}_2\Bigr)
\end{align}
on three registers (labelled $0, 1, 2$), where $\ket{\psi_i} = \frac{1}{\sqrt{2}}\left(\ket{0} \ket{\phi^0_i} + \ket{1} \ket{\phi^1_i}\right)$ for $i \in \{0,1\}$. Let the BQP machine also have access to $k, \tilde{t}, pk, evk$ as well as encryptions of $z_0, z_1$ under $pk$. Then, there is a BQP algorithm that computes a state that is within negligible trace distance from the state
\begin{align}\label{equation:threebit-goalstate}
    \frac{1}{\sqrt{2}}(\Z^{r_0} \otimes \Z^{r_1} \otimes \X^{r_2}) \Bigl((\ket{00}\ket{\phi^0_\mu} + \ket{11} \ket{\phi^1_\mu})_{2, \mu} \otimes \ket{\psi_{1-\mu}}_{1-\mu}\Bigr),
\end{align}
along with encryptions of $r_0, r_1, r_2$ under $pk$, with probability at least $1-\epsilon$.
\end{lemma}

\begin{proof}
First, we will introduce some notation that will be useful throughout the proof. Define the function $i: \{0, 1\}^3 \rightarrow \{0,1\}^2$ to be $i(u, v, w) := (u \oplus v, u \oplus w)$. Note that $i(u,v,w) = i(1-u, 1-v, 1-w)$. For readability, we will write $f_{i(u, v, w)}$ or $f_{uvw}$ to denote $f_{k, i(u, v, w)}$, depending on the context, where $k$ is the key given as an input. We will also abbreviate $D_{k,i(u,v,w)}$ as $D_{i(u, v, w)}$ or $D_{uvw}$.

We can rewrite the input state as
\begin{align*}
    \ket{\chi^{(1)}} =   (\Z^{z_0} \otimes \Z^{z_1} \otimes \I) \, \frac{1}{2 \sqrt{2}} \sum_{u, v, w \in \{0, 1\}} (\ket{v}\ket{\phi_0^v})_0 \otimes (\ket{w}\ket{\phi_1^w})_1 \otimes \ket{u}_2\;.
\end{align*}
By controlling appropriately on $u, v, w$, we can then prepare a state that is negligibly close in trace distance to the following state:
\begin{align*}
    & \ket{\chi^{(2)}} = (\Z^{z_0} \otimes \Z^{z_1} \otimes \I \otimes \I) \\
    & 
    \sum_{\substack{u, v, w \in \{0, 1\}\\ x \in \mathcal{X}} } \frac{1}{2 \sqrt{2}} \cdot \sqrt{D_{uvw}(x)} \cdot (\ket{v}\ket{\phi_0^v})_0 \otimes (\ket{w}\ket{\phi_1^w})_1 \otimes \ket{u}_2 \otimes \ket{x}_3.
\end{align*}
We now use the evaluation algorithm for $\mathcal{F}$ to coherently evaluate $f_{uvw}$ in register $4$.
\begin{align*}
    & \ket{\chi^{(3)}} = (\Z^{z_0} \otimes \Z^{z_1} \otimes \I \otimes \I \otimes \I)\\
    &\sum_{\substack{u, v, w \in \{0, 1\}\\ x \in \mathcal{X}} } \frac{\sqrt{D_{uvw}(x)}}{2 \sqrt{2}}  \cdot (\ket{v}\ket{\phi_0^v})_0 \otimes (\ket{w}\ket{\phi_1^w})_1 \otimes \ket{u}_2 \otimes \ket{x}_3 \otimes \ket{f_{uvw}(x)}_4.
\end{align*}
Now, we measure the fourth register to obtain some $y \in \mathcal{Y}$. At this point it is useful to introduce some more notation: let $\mathcal{X}_{b_1, b_2}$ be the set of all pre-images of $y$ under $f_{b_1, b_2}$, that is, $\mathcal{X}_{b_1, b_2} = \{x \mid y = f_{b_1, b_2}(x)\}$. Again, we will abuse notation and write $\mathcal{X}_{uvw}$ to denote $\mathcal{X}_{i(u, v, w)}$. Finally, let $\mathcal{Y}_{b_1, b_2}$ the image of $f_{b_1, b_2}$ when evaluated on the support of $D_{b_1, b_2}$.

For simplicity, let us first consider the case where $\mu = 0$. The case where $\mu = 1$ is symmetric.  Let $a \in \{0,1\}$ be such that $y \in \mathcal{Y}_{a, 0} \cup \mathcal{Y}_{a, 1}$. By the dual-mode property we know that there is exactly one such $a$. (Note that $a$ is not known to the evaluator in the clear at this point.) After the measurement of the fourth register, the superposition on the remaining registers collapses to have only values of $u, v, w$ and $x$ such that $y = f_{uvw}(x)$. In this case, since $y \in \mathcal{Y}_{a, 0} \cup \mathcal{Y}_{a, 1}$, we must have that $u \oplus v = a$. The resulting state is (up to normalization factors)
\begin{align*}
    \ket{\chi^{(4)}} & \propto (\Z^{z_0} \otimes \Z^{z_1} \otimes \I \otimes \I ) \cdot \\ &  \sum_{\substack{u, v, w \in \{0,1\}\\u \oplus v = a\\x \in \mathcal{X}_{uvw}}} \sqrt{D_{uvw}(x)} \cdot (\ket{v} \ket{\phi_0^v})_0 \otimes (\ket{w} \ket{\phi_1^w})_1 \otimes \ket{u}_2 \otimes \ket{x}_3\\
    &= (\Z^{z_0} \otimes \Z^{z_1} \otimes \I \otimes \I ) \cdot \\ 
    & \sum_{\substack{v, w \in \{0,1\}\\x \in \mathcal{X}_{a, a \oplus v \oplus w}}} \sqrt{D_{a, a \oplus v \oplus w}(x)} \cdot  (\ket{v} \ket{\phi_0^v})_0 \otimes (\ket{w} \ket{\phi_1^w})_1 \otimes \ket{v \oplus a}_2 \otimes \ket{x}_3.
\end{align*}
Rearranging the registers and writing $\ket{v \oplus a}$ as $\X^a \ket{v}$, we see that we almost have the desired state as in Equation~\ref{equation:threebit-goalstate}, except we want to uncompute register $3$,
\begin{align*}
    \ket{\chi^{(4)}} \propto & (\Z^{z_0} \otimes \Z^{z_1} \otimes \X^a \otimes \I)
    \cdot \\ & \sum_{\substack{v, w \in \{0,1\}\\x \in \mathcal{X}_{a, a \oplus v \oplus w}}} \sqrt{D_{a, a \oplus v \oplus w}(x)} \cdot  (\ket{v} \otimes \ket{v} \ket{\phi_0^v})_{2, 0} \otimes (\ket{w} \ket{\phi_1^w})_1 \otimes \ket{x}_3.
\end{align*}
Towards this end, we measure register $3$ in the Hadarmard basis. First applying the Hadamard transform on register $3$, we get a state proportional to
\begin{align*}
    & (\Z^{z_0} \otimes \Z^{z_1} \otimes \X^a \otimes \I) \cdot  \\
    & \sum_{\substack{v, w \in \{0,1\}\\x \in \mathcal{X}_{a, a \oplus v \oplus w}}} \sum_{d \in \bar{\mathcal{X}}} (-1)^{d \cdot x} \cdot \sqrt{D_{a, a \oplus v \oplus w}(x)} \cdot (\ket{v} \otimes \ket{v} \ket{\phi_0^v})_{2, 0} \otimes (\ket{w} \ket{\phi_1^w})_1 \otimes \ket{d}_3\;.
\end{align*}
Measuring register $3$, suppose we get outcome $d \in \bar{\mathcal{X}}$. The resulting state is then proportional to
\begin{align*}
   & \ket{\chi^{(5)}}  \propto  \,(\Z^{z_0} \otimes \Z^{z_1} \otimes \X^a \otimes \I) \\
   &\sum_{v, w \in \{0,1\}} \sum_{x \in \mathcal{X}_{a, a \oplus v \oplus w}} (-1)^{d \cdot x} \cdot \sqrt{D_{a, a \oplus v \oplus w}(x)} \cdot (\ket{v} \otimes \ket{v} \ket{\phi_0^v})_{2, 0} \otimes (\ket{w} \ket{\phi_1^w})_1.
\end{align*}
For convenience, let 
\begin{align*}
    \beta_{\nu} := \alpha_{k, y, d} (a, a \oplus \nu) = \sum_{x \in \mathcal{X}_{a, a \oplus \nu}} (-1)^{d \cdot x} \cdot \sqrt{D_{a, a \oplus \nu}(x)} \;, 
\end{align*}
so that we can rewrite the state as
\begin{align*}
    \ket{\chi^{(5)}} \propto (\Z^{z_0} \otimes \Z^{z_1} \otimes \X^a \otimes \I) \sum_{v, w \in \{0,1\}} \beta_{v \oplus w} \cdot (\ket{v} \otimes \ket{v} \ket{\phi_0^v})_{2, 0} \otimes (\ket{w} \ket{\phi_1^w})_1.
\end{align*}
We know that with probability at least $1-\epsilon$ over the measurement outcome $y$, $\beta_0 \in \{ \pm \beta_1 \}$. Let $s \in \{0,1\}$ be such that $\beta_1 = (-1)^s \cdot \beta_0$, we can factor out a global phase of $\beta_0$ and equivalently write the state as follows:
\begin{align*}
    \ket{\chi^{(5)}}  &\propto (\Z^{z_0} \otimes \Z^{z_1} \otimes \X^a \otimes \I) \sum_{v, w \in \{0,1\}} (-1)^{s \cdot (v \oplus w)} (\ket{v} \otimes \ket{v} \ket{\phi_0^v})_{2, 0} \otimes (\ket{w} \ket{\phi_1^w})_1\\
    &= (\Z^{z_0 \oplus s} \otimes \Z^{z_1 \oplus s} \otimes \X^a \otimes \I) \sum_{v, w \in \{0,1\}}  (\ket{v} \otimes \ket{v} \ket{\phi_0^v})_{2, 0} \otimes (\ket{w} \ket{\phi_1^w})_1.
\end{align*}
We can now homomorphically compute encryptions of the new one-time pad keys, since we are also given an encryption of the trapdoor. First, by the property of of efficient partial inversion, we can obtain all possible values of $(b_1, b_2)$ such that $y \in \mathcal{Y}_{b_1, b_2}$, and therefore obtain (an encryption of) the value of $r_2 = a$. Next, by the property of efficient phase computation, we can homomorphically compute $\beta_0, \beta_1$ and therefore $s$, which allows us to compute encryptions of $r_0 = z_0 \oplus s$ and $r_1 = z_1 \oplus s$.
\end{proof}

\subsection{Remotely Preparing the DSS Gadget}\label{section:dss-gadget-rsp-protocol}
We dequantize the DSS client by giving a protocol $\mathsf{GadgetKeyGen}(sk, pk')$ that allows the key generation to be classical, so that it outputs classical ``gadget keys'' $gk(sk)$ that depend on but computationally hide a secret key $sk$.
Using the gadget keys $gk$, the evaluator can then run $\mathsf{RSPGenGadget}$ to prepare a teleportation gadget state $\Gamma_{pk'}(sk)$ that can then be used as in \cite{dulek2016quantum} for the evaluation of $\T$ gates. The gadget key must also include trapdoor information encrypted under an independent public key $pk'$, so that the evaluator can homomorphically update its one-time pad key encryptions.

\subsubsection*{The client's procedure for generating the gadget keys.}
Suppose the length of the secret key is $|sk| = n$ bits. The client uses this procedure to ``encrypt'' the secret key $sk$ (in the form of keys of a ``$4$-to-$2$'' dual-mode trapdoor function family $\mathcal{F}$), which provide instructions to the evaluator to prepare teleportation gadget $\Gamma_{pk'}(sk)$. The procedure is as follows.\\

\noindent
$\underline{\mathsf{GadgetKeyGen}(sk, pk')}$
\begin{enumerate}
    \item For every $i \in [n]$, generate key-trapdoor pairs $k_i, t_i \leftarrow \KeyGen_\mathcal{F}(1^\secp, sk_i)$, giving the $i$th bit of $sk$ as input.
    \item Encrypt the trapdoors $t_i$ using the public key $pk'$, to get ciphertexts $\tilde{t}_i \leftarrow \HE.\Enc_{pk'}(t_i)$ for all $i \in [n]$.
    \item Output the gadget key $gk = (k_i, \tilde{t}_i)_{i=1}^n$.
\end{enumerate}

\subsubsection*{The evaluator's procedure for preparing the DSS gadget.} The evaluator uses the gadget keys generated by the client to prepare the DSS gadget. The procedure works as follows.\\

\noindent
$\underline{\mathsf{RSPGenGadget}(gk, \tilde{sk}, pk', evk')}$
\begin{enumerate}
    \item Parse the gadget key as $gk = (k_i, \tilde{t}_i)_{i=1}^n$.
    \item Prepare $|Q|$ qubit registers, each holding state $\ket{+}$.
    \item Maintain encryptions $\tilde{x}_q, \tilde{z}_q$ of the one-time pad keys for each register $q \in Q_1$ under $pk'$. Initialize all of these to encryptions of $0$.
    \item For every $q \in Q_1$, run the protocol described in Lemma~\ref{lemma:RSP-hidden-bell-pair} on registers $\pi_0(q)$, $\pi_1(q)$ and $q$ in the place of registers labeled $0$, $1$ and $2$ in the order. Let $i(q) \in [n]$ be such that the $i(q)$th bit $sk_{i(q)}$ of $sk$ determines the value of $\pi^{sk}(q)$ Also as input to the protocol to prepare the hidden Bell pair, provide $k_{i(q)}, \tilde{t}_{i(q)}, \tilde{sk}, pk', evk'$ and $\tilde{z}_{\pi_0(q)}$ and $\tilde{z}_{\pi_1(q)}$.
    \item As the output of the hidden Bell pair procedure, we get encryptions of the updated one-time pad keys, so we replace $\tilde{x}_q$, $\tilde{z}_{\pi_0(q)}$ and $\tilde{z}_{\pi_1(q)}$ with the updated keys.
    \item Finally, apply an inverse phase gate $\P^\dagger$ to the qubits in registers $q \in P$. Since $\P^\dagger$ is Clifford, simply update the encrypted one-time pad keys to reflect this change.
\end{enumerate}

\begin{lemma}\label{lemma:DSS-RSP-correctness}
Given a gadget key $gk \leftarrow \mathsf{GadgetKeyGen}(sk, pk')$, public key $pk'$ and evaluation key $evk'$, $\mathsf{RSPGenGadget}$ correctly outputs $\Gamma_{pk'}(sk) = (\gamma_{x, z}(sk) , \tilde{x}, \tilde{z})$, where
\begin{align}\label{equation:gamma-sk}
    \gamma_{x, z}(sk) := \bigotimes_{q \in Q_1} \X^{x_q} \Z^{z_q} (\P^\dagger)^{p_q} \ket{\Phi^+} \bra{\Phi^+}_{q, \pi^{sk}(q)}(\P)^{p_q} \Z^{z_q} \X^{x_q},
\end{align}
where $x, z \in \{0,1\}^{|Q_1|}$ are one-time pad keys and $\tilde{x}$ and $\tilde{z}$ are their encryptions under public key $pk'$.
\end{lemma}

\begin{proof}
We argue that throughout the $\mathsf{RSPGenGadget}(gk, pk', evk')$, the inputs to the procedure in Lemma~\ref{lemma:RSP-hidden-bell-pair}, to prepare hidden Bell pairs, are of the form of Equation~\eqref{equation:hidden-bell-pair-input}. As $\mathsf{RSPGenGadget}$ processes the registers $q \in Q_1$ one at a time, it prepares the corresponding component $\ket{\Phi^+} \bra{\Phi^+}_{q, \pi^{sk}(q)}$ of the tensor product state $\gamma_{x,z}(sk)$.
In the beginning, for the first register $q \in Q_1$, since all the registers are in the state $\ket{+}$, the input condition is satisfied, with $\ket{\psi_0} = \ket{\psi_1} = \ket{+}$. Let $\mu = sk_{i(q)}$, where $i(q)\in [n]$ is such that $sk_{i(q)}$ determines the value of $\pi^{sk}(q)$. After running the hidden Bell pair procedure, we obtain the state
\begin{align*}
    (\Z^{r_0} \otimes \Z^{r_1} \otimes \X^{r_2}) \left((\ket{00} + \ket{11})_{2, \mu} \otimes \ket{+}_{1-\mu}\right),
\end{align*}
as desired. For every subsequent $q \in Q_1$, the register $q$ is untouched thus far and therefore holds the $\ket{+}$ state. If $\mu = sk_{i(q)}$ and $q' = \pi^{sk}(q) \in Q_2$, we know that $q'$ cannot be already entangled to another qubit in $Q_1$ since $\pi^{sk}$ is injective. Therefore the state on register $\mu$ is $\Z^{z_\mu} \ket{\psi_\mu}$, where $\ket{\psi_\mu} = \ket{+}$. Further, the state on register $1-\mu$ must be in one of the following two states: either it was previously untouched, in which case it simply holds the state $\ket{+}$, or it is entangled with some other $\bar{q} \in Q_1$, in which case their joint state is of the form
\begin{align*}
    (\Z^{z_{1-\mu}} \otimes \X^{r}) \ket{\psi_{1-\mu}} = (\Z^{z_{1-\mu}} \otimes \X^{r}) \left( \ket{00} + \ket{11} \right)_{1-\mu, \bar{q}}\;,
\end{align*}
for some $z_{1-\mu}, r \in \{0,1\}$. Again, this satisfies the input conditions in Equation~\eqref{equation:hidden-bell-pair-input}. Therefore the procedure in Lemma~\ref{lemma:RSP-hidden-bell-pair} succeeds in correctly preparing the DSS gadget state (up to negligible trace distance error).
\end{proof}

\section{Our QFHE scheme}\label{section:our-scheme}
In this section, we give a quantum (leveled) fully homomorphic encryption scheme, built in a black-box way from compact {\emph classical} fully homomorphic encryption and dual-mode trapdoor functions.

\begin{theorem}\label{theorem:main-theorem}
    Let $\lambda$ be the security parameter. Suppose there exists a dual-mode trapdoor function family $\mathcal{F}$. Suppose further that there exists a leveled fully homomorphic encryption scheme $\mathsf{HE}$ with a decryption algorithm $\HE.\Dec$ in $\mathsf{NC}_1$, that is, $\HE.\Dec$ is computable by a boolean function fan-in $2$ circuit of depth $O(\log(\lambda))$. Then, there exists a compact quantum leveled fully homomorphic encryption scheme $\mathsf{QFHE}$ with a classical key generation procedure.
\end{theorem}
We will now describe our QFHE scheme.

\paragraph{\textbf{Key Generation.}}\label{section:key-generation}
Let $L$ be the maximum $\mathsf{T}$-depth of the circuits we want to homomorphically evaluate. Then, the client generates the keys as follows.\\

\noindent
$\underline{\mathsf{QFHE}.\mathsf{KeyGen}(1^\lambda, 1^L)}$
\begin{enumerate}
    \item For $i = 0$ to $L$, independently generate key $(pk_i, sk_i, evk_i) \leftarrow \HE.\KeyGen(1^\secp)$.
    \item Set the public key as $pk = (pk_i)_{i=0}^L$.
    \item Set the secret key as $sk = (sk_i)_{i=0}^L$.
    \item For $i = 0$ to $L-1$, compute $gk_i
    \leftarrow \mathsf{GadgetKeyGen}(sk_i, pk_{i+1})$. Compute an encryption of $sk_i$ under $pk_{i+1}$. Set the evaluation key $evk$ to be
    \begin{align*}
        \Bigl( gk_i,
        \HE.\Enc_{pk_{i+1}}(sk_i)\Bigr)_{i=0}^{L-1} , (evk_i)_{i=0}^{L}
    \end{align*}
\end{enumerate}

\paragraph{\textbf{Encryption.}}
The encryption procedure is identical to that of \cite{dulek2016quantum}. For completeness, we descibe it here.\\

\noindent
$\underline{\mathsf{QFHE}.\mathsf{Enc}_{pk}(\ket{\psi})}$
\begin{enumerate}
    \item Let $\ket{\psi}$ be an $n$-qubit input state. Sample uniformly random keys $x, z \leftarrow \{0, 1\}^n$. Parse the public key as $pk = (pk_i)_{i=0}^L$.
    \item Output the quantum one-time-padded state $\X^x \Z^z \ket{\psi}$ along with encryptions of $x, z$ under the first public key $pk_0$.
\end{enumerate}
Note that if the input is classical, $\ket{\psi}$ is just a standard basis state and the encryption is also classical.

\paragraph{\textbf{Circuit Evaluation.}}\label{section:circuit-evaluation}
The circuit evaluation procedure is almost the same as that of \cite{dulek2016quantum}, except that the evaluator needs to construct the teleportation gadget each time for each $\mathsf{T}$ gate it needs to perform. For completeness, we provide a description of the entire circuit evaluation procedure as in \cite{dulek2016quantum} with this modification. Without loss of generality, the quantum circuit $C$ we want to evaluate is written using a universal gate set consisting of the Clifford group generators $\{\mathsf{H}, \mathsf{P}, \mathsf{CNOT}\}$ and the (non-Clifford) $\mathsf{T}$ gate. Note that while the Clifford gates might affect multiple qubits at the same time, the $\T$ gates only affect a single qubit.
Before evaluation of a gate $U$, the encryption of the $n$-qubit state $\rho$ is of the form
\begin{align*}
    X^x Z^z \rho X^x Z^z = (\X^{x_1} \Z^{z_1} \otimes \cdots \otimes \X^{x_n} \Z^{z_n}) \rho (\X^{x_1} \Z^{z_1} \otimes \cdots \otimes \X^{x_n} \Z^{z_n})^\dagger.
\end{align*}
The evaluating party also holds encryptions $\tilde{x}_1^{[i]}, \ldots, \tilde{x}_n^{[i]}$ and $\tilde{z}_1^{[i]}, \ldots, \tilde{z}_n^{[i]}$ of the one-time pad keys, with respect to the $i$th key set $pk_i$ for some $i$ (initially, $i=0$). The goal is to produce an encrypted state $(\X^{x'}\Z^{z'}) U \rho U^\dagger (\X^{x'}\Z^{z'})^\dagger$, along with encryptions of the new keys $x', z'$ of the quantum one-time pad.
If $U$ is a Clifford gate, these encryptions will remain in the $i$th key set, but if $U$ is a $\T$ gate, all the encryptions are transferred to the $(i+1)$th key set.\\

\noindent
$\underline{\QFHE.\Eval(C, \rho, \tilde{x}^{[0]}, \tilde{z}^{[0]}, evk)}$
\begin{itemize}
    \item If $U$ is a Clifford gate, the evaluator simply applies $U$ to the encrypted qubit. Since $U$ commutes with the Pauli group, the evaluator only needs to update the encrypted keys in a straightforward way.
    \item If $U = \T$, the evaluator applies the $\T$ gate to the appropriate wire $w$. Supposing the qubit at wire $w$ starts off in the state $(\X^{x_w} \Z^{z_w}) \rho_w (\X^{x_w} \Z^{z_w})^\dagger$, after the $\T$ gate is applied, it is now in the state
    \begin{align*}
        (\P^{x_w} \X^{x_w} \Z^{z_w} \T) \rho_w (\P^{x_w} \X^{x_w} \Z^{z_w} \T)^\dagger.
    \end{align*}
    At this point, to correct the $\P^{x_w}$ error, Dulek et al.~\cite{dulek2016quantum} use the quantum gadget $\Gamma_{pk_{i+1}}(sk_i)$ that is given as input to the evaluator. In our scheme, the evaluator has to construct the quantum gadget itself from the classical evaluation key, using the evaluation key $gk_i, \tilde{t}^{[i+1]}_i, 
    \tilde{sk}^{[i+1]}_i, pk_{i+1}, evk_{i+1}$,
    \begin{align*}
        \Gamma_{pk_{i+1}}(sk_i) \leftarrow \mathsf{RSPGenGadget}
    \Bigl(gk_i, \tilde{t}^{[i+1]}_i, pk_{i+1}, evk_{i+1}\Bigr).
    \end{align*}
    The $\mathsf{RSPGenGadget}$ protocol is described later in this section. Now, as in DSS, the evaluator performs a sequence of measurements as specified by the function $\nu(\tilde{x}_w^{[i]}, \cdot)$. This has the effect of the qubit at wire $w$ and 
    \begin{align*}
        (\X^{x'_w} \Z^{z'_w}) \T \rho \T^\dagger (\X^{x'_w} \Z^{z'_w})^\dagger.
    \end{align*}
    The evaluator also recrypts $\tilde{x}^{[i]}_w$ and $\tilde{y}_w^{[i]}$ under the $(i+1)$th key, to get $\tilde{x}^{[i+1]}_w$ and $\tilde{y}_w^{[i+1]}$. Using the measurement outcomes, the classical information part of the gadget $\Gamma_{pk_{i+1}}(sk_i)$ and the recryptions $\tilde{x}^{[i+1]}_w$, $\tilde{y}_w^{[i+1]}$ the evaluator homomorphically computes encryptions of the new keys $(\tilde{x}'_w)^{[i+1]}$, $(\tilde{y}'_w)^{[i+1]}$. After these computations, the evaluator also recrypts the keys of all the wires into the $(i+1)$th key set.
\end{itemize}

\paragraph{\textbf{Decryption.}}
As input, we are given the state $\ket{\psi'}$, which is hopefully close to the state $\X^x \Z^z C\ket{\psi}$, where $C$ is the circuit and $C\ket{\psi}$ is the desired output state, except it is quantum one-time-padded with keys $x, z \in \{0, 1\}^n$. We are also given encryptions $\tilde{x}, \tilde{z}$ of the one-time-pad keys $x, z$ under $pk_L$.\\

\noindent
$\underline{\QFHE.\Dec_{sk}(\ket{\psi'}, \tilde{x}, \tilde{z})}$
\begin{enumerate}
    \item Decrypt encryptions of the one-time-pad keys using $sk_L$ to get $x, z$.
    \item Apply the Pauli correction $\X^x \Z^z$ to $\ket{\psi'}$ and output the resulting state.
\end{enumerate}

\subsection{Proof of Theorem~\ref{theorem:main-theorem}}

We first prove the quantum IND-CPA security of our scheme by constructing a sequence of hybrid schemes that transforms our scheme to the scheme by Broadbent and Jeffery~\cite{broadbent2015quantum}, which we will denote as $\mathsf{BJ}$.

\begin{lemma}\label{lemma:security}
Suppose that $\HE$ is a IND-CPA-secure classical homomorphic encryption scheme. Then, $\QFHE$ is a quantum IND-CPA secure quantum homomorphic encryption scheme.
\end{lemma}

\begin{proof}
Consider the following sequence of hybrid variations of our scheme. For all $\ell \in \{0,1,\ldots, L\}$, let $\mathsf{S}_\ell$ be identical to our scheme except for the key generation procedure: for every $i \ge \ell$, $\mathsf{S}_\ell.\mathsf{KeyGen}$ replaces the output $gk_i(sk_i)$, $\mathsf{HE}.\mathsf{Enc}_{pk_{i+1}}(t_{gk_i})$ of the subroutine call $\mathsf{GadgetKeyGen}(sk_i, pk_{i+1})$ as follows. The the gadget key $gk_i(sk_i)$ encoding the $i$th secret key is replaced with a gadget key $gk(0^{|sk_i|})$ encoding the all-zero string obtained by $$gk(0^{|sk_i|}), \tilde{t} \leftarrow \mathsf{GadgetKeyGen(0^{|sk_i|}, pk_{i+1})}~.$$ The encryption of $t_{gk_{i}}$ under the public key $pk_{i+1}$ is replaced with an encryption $\mathsf{HE}.\mathsf{Enc}_{pk_{i+1}}(0^{|t_{gk_{i}}|})$ of the all-zero string $0^{|t_{gk_{i}}|}$ under $pk_{i+1}$.

Also, for every $\ell \in \{0,1 \ldots L-1\}$, define scheme $\mathsf{S}'_\ell$ be identical to $\mathsf{S}_{\ell+1}$, except that the encryption $\mathsf{HE}.\mathsf{Enc}_{pk_{i+1}}(t_{gk_i})$ of the $i$th trapdoor under $pk_{i+1}$ is replaced with an encryption of the all-zero string $0^{|t_{gk_{i}}|}$ under the same public key.

First, we claim that for any adversary $\mathcal{A}$, for all $\ell \in \{0,1,\ldots,L-1\}$, there exists a negligible function $\negl$ such that
\begin{align*}
    \Pr[\QINDCPA_{\langle \mathsf{S}_{\ell + 1}, \mathcal{A} \rangle}] - \Pr[\QINDCPA_{\langle \mathsf{S}'_\ell, \mathcal{A}\rangle}] \le \negl(\lambda).
\end{align*}
Otherwise, we can construct an adversary $\mathcal{A}'_{\mathsf{HE}}$ that breaks the IND-CPA security of the classical $\mathsf{HE}$ scheme by simulating the $\QINDCPA_{\langle \rangle}$ game for $\mathcal{A}$.

Similarly, for any adversary $\mathcal{A}$, it must be the case that for all $\ell \in \{0,1,\ldots, L-1\}$, there exists a negligible function $\negl$ such that
\begin{align*}
    \Pr[\QINDCPA_{\langle \mathsf{S}'_\ell, \mathcal{A}\rangle}] - \Pr[\QINDCPA_{\langle \mathsf{S}_{\ell}, \mathcal{A} \rangle}] \le \negl(\lambda).
\end{align*}
Otherwise, we can construct an adversary $\mathcal{A}'_{\mathcal{F}}$ that breaks the mode indistinguishability property of the dual-mode trapdoor function family $\mathcal{F}$ by simulating $\QINDCPA$ game for $\mathcal{A}$.

Note that $\mathsf{S}_L$ is exactly our scheme $\mathsf{QFHE}$, and $\mathsf{S}_0$ gives the adversary no extra information compared to the $\mathsf{BJ}$ scheme, so that for all adversaries $\mathcal{A}$,
\begin{align*}
    \Pr[\QINDCPA_{\langle\mathsf{S}_L, \mathcal{A} \rangle}(1^\lambda)] &= \Pr[\QINDCPA_{\langle\mathsf{QFHE}, \mathcal{A} \rangle}(1^\lambda)], \text{ and}\\
    \Pr[\QINDCPA_{\langle\mathsf{S}_0, \mathcal{A} \rangle}(1^\lambda)] &\le \Pr[\QINDCPA_{\langle \mathsf{BJ} , \mathcal{A} \rangle}(1^\lambda)]
\end{align*}
\noindent
Therefore,
\begin{align}\label{equation:qfhe-bj}
    \Pr[\QINDCPA_{\langle\mathsf{QFHE}, \mathcal{A} \rangle}(1^\lambda)] &- \Pr[\QINDCPA_{\langle \mathsf{BJ} , \mathcal{A} \rangle}(1^\lambda)] \le \Pr[\QINDCPA_{\langle\mathsf{S}_L, \mathcal{A} \rangle}(1^\lambda)] - \Pr[\QINDCPA_{\langle\mathsf{S}_0, \mathcal{A} \rangle}(1^\lambda)]
\end{align}
We can write the right-hand side as a telescoping series
\begin{align*}
     \Pr[\QINDCPA_{\langle\mathsf{S}_L, \mathcal{A} \rangle}(1^\lambda)] - \Pr[\QINDCPA_{\langle\mathsf{S}_0, \mathcal{A} \rangle}(1^\lambda)]
     = \sum_{\ell = 0}^{L-1} &\left(\Pr[\QINDCPA_{\langle\mathsf{S}_{\ell + 1}, \mathcal{A} \rangle}(1^\lambda)] - \Pr[\QINDCPA_{\langle\mathsf{S}'_{\ell}, \mathcal{A} \rangle}(1^\lambda)] \right) \\ 
    & + \left(\Pr[\QINDCPA_{\langle\mathsf{S}'_{\ell}, \mathcal{A} \rangle}(1^\lambda)] - \Pr[\QINDCPA_{\langle\mathsf{S}_{\ell}, \mathcal{A} \rangle}(1^\lambda)]\right)\\
    &\le 2 L \cdot \negl(\lambda) \le \negl'(\lambda),
\end{align*}
for some other negligible function $\negl'(\cdot)$.
By Equation~\eqref{equation:qfhe-bj}, we know that there exists a neglible function $\negl'$ such that
\begin{align*}
    \Pr[\QINDCPA_{\langle\mathsf{QFHE}, \mathcal{A} \rangle}(1^\lambda)] - \Pr[\QINDCPA_{\langle \mathsf{BJ} , \mathcal{A} \rangle}(1^\lambda)] \le \negl'(\lambda).
\end{align*}
By the quantum IND-CPA security of the $\mathsf{BJ}$ scheme, we get that no quantum polynomial-time adversaries $\mathcal{A}$ can win the quantum IND-CPA game $\QINDCPA_{\langle \mathsf{QFHE}, \cdot \rangle}$ with respect to our scheme $\mathsf{QFHE}$ with more than negligible advantage.
\end{proof}

\begin{proof}[Proof of Theorem~\ref{theorem:main-theorem}]
Lemma~\ref{lemma:security} shows that the scheme $\QFHE$ is quantum IND-CPA secure. Lemma~\ref{lemma:DSS-RSP-correctness} shows the correctness of the protocol for preparing the DSS gadget on the evaluator side. Therefore, since the DSS scheme is correct and leveled fully homomorphic, so is ours. Finally, since our decryption scheme (which is identical to that of DSS) does not depend on the evaluated circuit, $\QFHE$ is compact.
\end{proof}

\section{Instantiations}\label{section:instantiations}
Our scheme in Section~\ref{section:our-scheme} allows us to build QFHE from generic classical FHE and dual-mode trapdoor functions. We now list the known instantiations of these primitives from various (concrete) cryptographic assumptions.

\ifnum\llncs=0
\subsection{Instantiating post-quantum classical (leveled) FHE}
\fi

\paragraph{Learning With Errors.}
The Learning With Errors assumption~\cite{regev2009lattices} can be used to contruct both (leveled) fully homomorphic encryption as well as dual-mode trapdoor functions.
Extensive research has been done on constructing FHE from LWE, culminating in the work of Brakerski and Vaikuntanathan~\cite{brakerski2014efficient} that builds FHE from the standard LWE assumption with polynomial modulus-to-noise ratio.

\ifnum\llncs=1
Mahadev~\cite{mahadev2018classical} constructs a dual-mode noisy trapdoor claw-free function family assuming the hardness of Learning With Errors (with superpolynomial modulus). This construction also immediately gives a dual-mode trapdoor function family from the hardness of LWE for superpolynomial modulus-to-noise ratio. In the full version, we give a description of this construction.

We also observe that we can also obtain (a $q$-ary generalization of) dual-mode trapdoor functions from the quantum hardness of LWE with a polynomial modulus-to-noise ratio from the work of Brakerski~\cite{DBLP:conf/crypto/Brakerski18}, which builds a QFHE scheme from the same assumption. The full version of this paper contains a brief description of the construction.
\fi

\paragraph{Ring Learning With Errors.}
Brakerski, Gentry and Vaikuntanathan also give an FHE scheme from Ring-LWE~\cite{DBLP:journals/toct/BrakerskiGV14}, which can be directly imported into our QFHE scheme.
\ifnum\llncs=1
The dual-mode trapdoor function family from LWE of \cite{mahadev2018classical} can be directly translated to one from Ring-LWE. We formally describe such a construction from ring-LWE in the full version.
\fi

\paragraph*{Post-Quantum FHE from Post-quantum Indistinguishability Obfuscation + perfectly re-randomizable encryption.}
The work of Canetti, Lin, Tessaro and Vaikuntanathan~\cite{canetti2015obfuscation} obtains leveled classical FHE from subexponentially-secure IO and perfectly re-rerandomizable encryption. Their scheme is post-quantum as long as it is instantiated with subexponentially-secure post-quantum IO and post-quantum secure perfectly re-randomizable encryption. Although we have candidates for post-quantum IO~\cite{gentry2015graph, bartusek2018return, brakerski2020factoring, wee2021candidate}, we do not have proofs of their post-quantum security under widely-believed assumptions. On the other hand, by an observation of Wichs~\cite{Wichs}, there exists an El Gamal-style perfectly rerandomizable encryption scheme based on group actions.

\ifnum\llncs=0
\subsection{Instantiating dTLFs}
\fi

\paragraph*{Post-Quantum dTF from Isogeny-based group action assumptions.}
Isogeny-based cryptography is a promising alternative to lattices for post-quantum cryptography. Initiated by Couveignes~\cite{couveignes2006hard} in 1997, it has led to new constructions of several cryptographic primitives such as collision-resistant hashing~\cite{charles2009cryptographic}, key exchange~\cite{rostovtsev2006public, stolbunov2010constructing}, digital signatures~\cite{stolbunov2009reductionist} and key escrow~\cite{teske2006elliptic}. More recent work has led to practical key exchange and public key encryption schemes, such as SIDH~\cite{jao2011towards, de2014towards} and CSIDH~\cite{castryck2018csidh}. Although the attack by Castryck and Decru~\cite{castryck2023efficient} (classically) breaks SIDH, CSIDH is still plausibly (post-quantum) secure.

Alamati, De Feo, Montgomery and Patranabis~\cite{alamati2020cryptographic} present a framework based on group actions with natural cryptographic hardness properties to model isogeny-based assumptions and schemes.\footnote{An important caveat is that the isogeny-based group actions have less-than-ideal algorithmic properties; for example, it is not always possible to efficiently compute the group action for any group element. One approach to fix this issue, taken by CSI-FiSh~\cite{beullens2019csi}, is to perform a preprocessing step and compute the group structure in the form of relation lattice of low-norm generators.} This simplifying framework enables them to construct several new primitives, such as projective hashing, dual-mode PKE, two-message statistically sender-private OT, and a Naor-Reingold style PRF. They also introduce the Linear Hidden Shift assumption, based on which they construct symmetric KDM-secure encryption. Follow-up work of Alamati, Malavolta and Rahimi~\cite{alamati2022candidate} introduces the \textit{Extended} Linear Hidden Shift assumption, and gives a construction of a (weak) trapdoor claw-free function from group actions under this assumption.

In Section~\ref{sec:groupactions}, we extend their work to construct dual-mode trapdoor functions (with negligible correctness error) from the same assumption.

\ifnum\llncs=0
\subsubsection*{Learning With Errors (Superpolynomial modulus-to-noise ratio)}
Mahadev~\cite{mahadev2018classical} constructs a dual-mode noisy trapdoor claw-free function family assuming the hardness of Learning With Errors (although with superpolynomial modulus). We slightly adapt this construction to match our definition of dTFs. Assuming the hardness of LWE for superpolynomial modulus-to-noise ratio, we obtain a dTF family.
\begin{lemma}
Let $\lambda$ be a security parameter. Let $q \ge 2$ be a prime. Let $n, m\ge 1$ be polynomially bounded functions of $\lambda$ such that $m = \Omega(n \log q)$. Suppose that $\LWE_{m, n, q, \rho_{\mathbb{Z}_{q}, \sigma}}$ is hard for superpolynomial modulus-to-noise ratio $q/\sigma$. Then, there exists an injective and invertible dual-mode trapdoor function family.
\end{lemma}

\begin{proof} Let $B^* = q/(C\sqrt{m n \log q})$, where $C > 0$ is the universal constant as in Theorem~\ref{theorem:mp12}. Let $B_1 = B^*/2$ and let $B_2 < B_1$ such that the ratio $B_1 / B_2$ is superpolynomial in the security parameter $\lambda$.

\paragraph{Key Generation.} The key $k = (\matA, \vecb) \in \mathbb{Z}_q^{m \times n} \times \mathbb{Z}_q^{m}$ and trapdoor $t_k$ is generated as follows. First $\mat{A} \in \mathbb{Z}_q^{m \times n}$ is generated along with a trapdoor $t_k = t_{\mat{A}}$ as in Theorem~\ref{theorem:mp12}~\cite{micciancio2012trapdoors}. In mode $\mu = 0$, $\vecb \leftarrow \mathbb{Z}_q^m$ is sampled uniformly at random. In mode mode $\mu = 1$, $\vec{s} \leftarrow \mathbb{Z}_q^n$ is sampled uniformly, and $\bar{\vece} \leftarrow D_{\mathbb{Z}_q^m, B_2}$ and $\vecb = \matA \vecs + \bar{\vece}$.

\paragraph{Efficient Function Evaluation.} The function $f_{k, b} : \mathbb{Z}_q^n \times \in \mathbb{Z}_q^m \rightarrow \mathbb{Z}_q^m$ takes as input $\vecx \in \mathbb{Z}_q^m$ and $\vece \in \mathbb{Z}_q^m$ such that $\|\vece\| \le B^* \sqrt{m}$, and is defined as $f_{k, b}(\vecx, \vece) := \matA \vecx + b \cdot \vecb + \vece$, which is efficiently computable given the key $k = (\matA, \vecb)$.

\paragraph{Efficient state preparation.} Let $D = U(\mathbb{Z}^n_q) \otimes \rho_{\mathbb{Z}_q^m, B_1}$ be the product distribution of the uniform distribution on $\mathbb{Z}^n_q$ and the truncated discrete Gaussian distribution $\rho_{\mathbb{Z}_q^m, B_1}$. Let $D' = U(\mathbb{Z}^n_q) \otimes (\rho_{\mathbb{Z}_q^m, B_1} + \bar{\vece})$, where $\rho_{\mathbb{Z}_q^m, B_1} + \bar{\vece}$ is the {\em shifted} truncated discrete Gaussian distribution. By Lemma~3.12 of \cite{regev2009lattices}, there is a quantum polynomial-time procedure that prepares the state 
\begin{align*}
\sum_{\vecx \in \mathbb{Z}_q^n} \sum_{\vece \in \mathbb{Z}_q^m} \sqrt{D(\vecx, \vece)} \ket{\vecx, \vece} = \sum_{\vecx \in \mathbb{Z}_q^n} \sum_{\vece \in \mathbb{Z}_q^m} \frac{1}{q^{-n/2}} \cdot \sqrt{\rho_{\mathbb{Z}_q^m, B_1}(\vece)} \ket{\vecx, \vece}\;,
\end{align*}

In mode $\mu = 0$, define $D_{k, 0} = D_{k, 1} = D$. In mode $\mu = 1$, for key $k = (\matA, \vecb = \matA \vecs + \bar{\vece})$, let $D_{k, 0} = D$ and let $D_{k, 1} = D'$. Although $\bar{\vece}$ is not efficiently computable without the trapdoor, and therefore we cannot work with $D'$ exactly, $D$ is very close to $D'$, and that will be good enough. Indeed since $\|\bar{\vece}\| \le B_2 \sqrt{m}$, by Lemma~\ref{lemma:hellinger-shifted-discrete-gaussian} and the parameter setting, $H(D, D') \le \negl(\lambda)$, and by Lemma~\ref{lemma:hellinger-trace-distance}, $\sum_{\vecx \in \mathbb{Z}_q^n} \sum_{\vece \in \mathbb{Z}_q^m} \sqrt{D(\vecx, \vece)} \ket{\vecx, \vece}$ and $\sum_{\vecx \in \mathbb{Z}_q^n} \sum_{\vece \in \mathbb{Z}_q^m} \sqrt{D'(\vecx, \vece)} \ket{\vecx, \vece}$ are negligibly close in trace distance.

\paragraph{Injectivity and Efficient inversion with the trapdoor.} Given the trapdoor and some $\vecy = \matA \vecx + b \cdot \vecb + \vece$ such that $\| \vece\| \le B_1 \sqrt{m}$, we can use Theorem~\ref{theorem:mp12} for each $b \in \{0,1\}$ on input $(\matA, \vecy - b \cdot \vecb)$ to recover $\vecx_b$ if it exists. This will give us all the possible preimages. This also immediately implies that $f_{k, b}$ is injective.

\paragraph{Dual-mode.} In mode $\mu = 0$, with all but negligible probability there does not exist $\vecs, \bar{\vece}$ such that $\vecb = \matA \vecs + \bar{\vece}$ such that $\|\bar{\vece}\| \le B^* \sqrt{m}$. Therefore there cannot exist $(\vecx_0, \vece_0)$ and $(\vecx_1, \vece_1)$ in the supports of $D_{k, 0}$ and $D_{k, 1}$ respectively such that $f_{k, 0}(\vecx_0, \vece_0) = f_{k, 1}(\vecx_1, \vece_1)$; otherwise, $\matA \vecx_0 + \vece_0 = \matA \vecx_1 + \vecb + \vece_1$. Rearranging, this means that $\vecb = \matA (\vecx_0 - \vecx_1) + (\vece_0 - \vece_1)$, where $\|\vece_0 - \vece_1\| \le \|\vece_0\| + \|\vece_1\| \le 2 B_1 \sqrt{m} = B^* \sqrt{m}$, which is a contradiction.

In mode $\mu = 1$, $\vecb = \matA \vecs + \bar{\vece}$ for $\|\bar{\vece}\|\le B_2 \sqrt{m}$, by the injectivity of $f_{k, b}$, there is a perfect matching $\mathcal{R}$ such that $((\vecx_0, \vece_0), (\vecx_1, \vece_1)) \in \mathcal{R}$ if and only if $f_{k, 0}(\vecx_0, \vece_0) = f_{k, 1}(\vecx_1, \vece_1)$. In particular, $f_{k, 0}(\vecx_0, \vece_0) = f_{k, 1}(\vecx_1, \vece_1)$ if and only if $\vecx_1 = \vecx_0 - \vecs$ and $\vece_1 = \vece_0 - \bar{\vece}$. Therefore, $D_{k, 0}(\vecx_0, \vece_0) = D_{k, 1}(\vecx_1, \vece_1)$.

\paragraph{Mode indistinguishability.} Assuming that the $\LWE_{n, q, \rho_{\mathbb{Z}_q, \sigma}}$ assumption is true for some superpolynomial modulus-to-noise ratio $B_2/q$, the keys in the two modes are computationally indistinguishable.
\end{proof}

\fi

\ifnum\llncs=0
\subsubsection*{Ring Learning With Errors}
The dual-mode trapdoor function family from LWE described above can be directly translated to one from Ring-LWE. We now informally describe such a construction from ring-LWE. The \cite{micciancio2012trapdoors} trapdoor sampling procedure readily generalizes to the Ring-LWE setting. Let $R_q := \mathbb{Z}_q[x]/(x^d + 1)$ be a polynomial ring.

The key consists of two sequences of polynomials $k = (\veca, \vecu) \in R_q^m \times R_q^m$. The sequence of polynomials $\veca = (a_1, a_2, \ldots, a_m) \in R^m_q$ generated along with a trapdoor $t_{\vec{a}}$. In mode $\mu = 0$, $\vecu = (u_1, u_2, \ldots, u_m) \in R^m_q$ generated uniformly at random. In mode $\mu = 1$, $\vecu = \veca \cdot s + \bar{\vece}$, where $s \in R_q$ is any (worst-case) polynomial and $\bar{\vece} \in R_q^m$ is a vector of polynomials whose each coefficient is sampled i.i.d.~from a narrow-width discrete Gaussian. Then the function $f_{k, b}$ is defined as follows. For $x \in R_q$ and $\vece \in R_q^m$, 
\begin{align*}
    f_{k, b}(x, \vece) := \veca \cdot x + b \cdot \vecu + \vece.
\end{align*}
Similar to the LWE case, in mode $\mu=0$ the probability distributions $D_{k, 0} = D_{k, 1}$ are (wider) discrete Gaussians, and in mode $\mu=1$, and the two distributions are discrete Gaussians of the same width, except that they are shifted by $\bar{\vece}$, that is $D_{k, 1} = D_{k, 0}-\bar{\vece}$.
\fi

\ifnum\llncs=0
\subsubsection*{Learning With Errors (Polynomial modulus-to-noise ratio)}
We observe that we can also obtain (a $q$-ary generalization of) dual-mode trapdoor functions from the quantum hardness of LWE with a polynomial modulus-to-noise ratio from the work of Brakerski~\cite{DBLP:conf/crypto/Brakerski18}, which builds a QFHE scheme from the same assumption. Unlike the dTF obtained from LWE with superpolynomial modulus-to-noise ratio, the functions we obtain here are not injective; however, there is still enough structure that allows us to correct the phase error introduced after the Hadamard measurement of the input register. We note that this $q$-ary variant of dTFs suffice for our purposes; instead of measuring the third register of $\ket{\chi^{(4)}}$ in the Hadamard basis in Lemma~\ref{lemma:RSP-hidden-bell-pair}, we instead first take the $q$-ary QFT and then measure the third register in the standard basis.

The \cite{DBLP:conf/crypto/Brakerski18} QFHE scheme uses the work of Bourse et al.~\cite{bourse2016fhe} that obtains circuit-private classical FHE from LWE with a polynomial modulus-to-noise ratio. In particular, we need a classical encryption scheme $\HE_{\text{Bra18}}$ that supports a rerandomizable, circuit-private homomorphic scalar addition operation, with the following additional properties. The QFHE scheme of Brakerski~\cite{DBLP:conf/crypto/Brakerski18} implicitly relies on such an encryption, and instantiates it assuming the hardness of LWE with polynomial modulus-to-noise ratio. We informally describe the dTF construction and the properties we need from the encryption scheme.

Let $P_0, P_1$ be some distributions over the ciphertext space such that ciphertexts in the support of $P_b$ decrypt to $b\in\{0,1\}$, so $\TV(P_0, P_1) = 1$ by correctness of decryption. Further, let $D$ be an efficiently sampleable distribution over finite domain $\mathcal{X}$.

\begin{enumerate}
    \item There is a probabilistic polynomial-time key generation algorithm $pk, sk, evk \leftarrow \HE_{\text{Bra18}}.\KeyGen(1^\lambda)$.
    \item The encryption algorithm $c \leftarrow \HE_{\text{Bra18}}.\Enc_{pk}(m; r)$ encrypts bits $m\in\{0,1\}$ with the public key. The ciphertext space is contained in $\mathbb{Z}_q^n$ for some $q, n \in \mathbb{N}$.
    \item There exists an algorithm that computes scalar addition as follows: $\HE_{\text{Bra18}}.\mathsf{EvalAdd}_{evk}(c, b; x)$ takes as input an encryption $c$ of $m\in\{0,1\}$, a bit $b\in\{0,1\}$ and randomness $x \leftarrow D$, and outputs ciphertext $c'$ such that the distribution of $c'$ (taken over the choice of $x \leftarrow D$) is negligibly close in Hellinger distance to $P_{m \oplus b}$.
    \item There is an efficient quantum procedure that produces the state negligibly close to $\sum_{x\in\mathcal{X}}\sqrt{D(x)} \ket{x}$.
    \item There is a decryption algorithm $\HE_{\text{Bra18}}.\Dec_{sk}$ that correctly decrypts the ciphertext outputs of $\HE_{\text{Bra18}}.\Enc$ and $\HE_{\text{Bra18}}.\mathsf{EvalAdd}$.
\end{enumerate}

With such an encryption scheme $\HE_{\text{Bra18}}$, we can construct dTF as follows. The key generation procedure takes as input the security parameter $1^\lambda$ and the mode $\mu\in\{0,1\}$, and outputs the key-trapdoor pair $k = (k_0, k_1, evk)$ and $t_k = sk$, where $(pk, sk, evk) \leftarrow \HE_{\text{Bra18}}.\KeyGen(1^\lambda)$, $k_0 \leftarrow \HE_{\text{Bra18}}.\Enc_{pk}(0)$ and $k_1 \leftarrow \HE_{\text{Bra18}}.\Enc_{pk}(\mu)$. By the semantic security of $\HE_{\text{Bra18}}$, keys sampled from the two modes are computationally indistinguishable. The dTF is defined as
\begin{align*}
    f_{k, b}(x) = \HE_{\text{Bra18}}.\mathsf{EvalAdd}_{evk}(k_b, b; x),
\end{align*}
which is efficiently computable given the key $k = (k_0, k_1)$. In the disjoint mode ($\mu = 0$), the output of $f_{k, b}$ is always an encryption of the bit $b$, so $f_{k, 0}$ and $f_{k, 1}$ have disjoint images. In the lossy mode ($\mu = 1$), both $f_{k, 0}$ and $f_{k, 1}$ always output encryptions of $0$. However these functions are no longer injective. To ensure the lossy mode property and efficient phase computation, we need the following additional property from the encryption scheme $\HE_{\text{Bra18}}$.
\begin{enumerate}[resume]
    \item For all ciphertext outputs $y \in \mathbb{Z}_q^n$ of $\HE_{\text{Bra18}}.\mathsf{EvalAdd}$, and for all $d \in \mathbb{Z}^n_q$, it must be true that there is some $s \in \{0,1\}$ such that $\alpha_0 = (-1)^s \alpha_1$, where
    \begin{align*}
        \alpha_b = \sum_{\substack{x \in \mathcal{X}\\f_{k, b}(x) = y}} e^{\frac{-2\pi i}{q} \langle d \cdot x\rangle} \cdot \sqrt{D(x)}.
    \end{align*}
    Further, given $y, d$ as well as the secret key/trapdoor as input, there is an efficient algorithm that computes $s$ for any $y, d$.
\end{enumerate}
\fi

\section{Dual-Mode Trapdoor Functions from Group Actions}\label{sec:groupactions}
Alamati, Malavolta and Rahimi give a candidate $\epsilon$-weak trapdoor function from isogeny-based group actions~\cite{alamati2022candidate} for inverse-polynomial $\epsilon$. For the security of their construction, they introduce and rely on the Extended Linear Hidden Shift (ELHS) assumption, a strengthening of the Linear Hidden Shift assumption of Alamati et al.~\cite{alamati2020cryptographic}. We modify their construction to give an \emph{injective and invertible dual-mode trapdoor  function} based on the ELHS assumption. Although this initially gives us a $\epsilon$-weak dTF with the same inverse polynomial correctness error, Lemma~\ref{lemma:dTF-amplification} immediately implies a dTF with negligible correctness error. We start by setting up some notation and definitions.

\paragraph{Notation.} For a regular and abelian group action $\star: \mathbb{G} \times \mathbb{X} \rightarrow \mathbb{X}$, we use additive notation to denote the group operation in $\mathbb{G}$. Let $\odot$ denote the component-wise product.

\begin{definition}[Effective Group Action] A regular and abelian group action $(\mathbb{G}, \mathbb{X}, \star)$ is effective if it satisfies the following properties.
\begin{enumerate}
    \item The group $\mathbb{G}$ is finite and there exist efficient p.p.t.~algorithms for:
    \begin{enumerate}
        \item membership testing (deciding whether a binary string represents a group element),
        \item equality testing and sampling uniformly in $\mathbb{G}$, and
        \item group operation and computing the inverse of any element.
    \end{enumerate}
    \item The set $\mathbb{X}$ is finite and there exist efficient algorithms for
    \begin{enumerate}
        \item membership testing (to check if a binary string represents a valid set element), and
        \item unique representation.
    \end{enumerate}
    \item There exists a distinguished element $x_0 \in \mathbb{X}$ with known representation.
    \item There exists an efficient algorithm that given any $g \in \mathbb{G}$ and any $x \in \mathbb{X}$, outputs $g \star x$.
\end{enumerate}
\end{definition}

Isogeny-based group actions are believed to have cryptographic properties such as one-wayness and pseudorandomness, suitably defined. One such conjectured hardness property is captured by the Extended Linear Hidden Shift assumption, introduced by Alamati, Malavolta and Rahimi~\cite{alamati2022candidate}. We note that a single-copy version of the assumption as stated in \cite{alamati2022candidate} suffices for us. Alamati et al.~\cite{alamati2022candidate} need (something stronger than) the multiple-copy version, to construct a trapdoor claw-free function with the adaptive hardcore bit property, which we do not require. We now define the single-copy version of their assumption.

\begin{definition}[Single-copy Extended Linear Hidden Shift Assumption \cite{alamati2022candidate}]\label{def:elhs} Let $\lambda$ be the security parameter. Let $(\mathbb{G}, \mathbb{X}, \star)$ be an effective group action, and let $n \ge \log (|\mathbb{G}|) + \Omega (\poly(\lambda))$ be an integer. We say that the extended linear hidden shift (LHS) assumption holds over $(\mathbb{G}, \mathbb{X}, \star)$ if no quantum polynomial-time adversary can distinguish between the following two distributions with advantage greater than negligible in $\lambda$:
\begin{align}\label{equation:elhs-distributions}
    D_0 := \left(\matM, \vecm, (\veca_b, \vecb_b)_{b \in \{0,1\}}\right) \quad \text{and} \quad
    D_1 := \left(\matM, \vecm, (\vecx_b, \vecy_b)_{b \in \{0,1\}}\right)\;,
\end{align}
where $\matM \leftarrow \mathbb{G}^{n \times n}$, $\vecm \leftarrow \mathbb{G}^n$, $\vecs \leftarrow \{0,1\}^n$, $\vect \leftarrow \mathbb{G}^n$, $\veca_b \leftarrow \mathbb{X}^n$, and $\vecb_b \leftarrow \mathbb{X}^n$ for $b \in \{0, 1\}$, and
\begin{align*}
    \vecx_0 &\leftarrow \mathbb{X}^n,\\
    \vecy_0 &:= \vect \star \vecx_0,\\
    \vecx_1 &:= \left[\matM \vecs \right] \star \vecx_0,\\
    \vecy_1 &:= \left[ \matM \vecs + \vecm \odot \vecs \right] \star \vecy_0\;.
\end{align*}
\end{definition}

We make the following assumption regarding the structure of the group $\mathbb{G}$. This allows us to reason about the injectivity of the dTF family we construct. Informally, we require that the group order $\ord(\mathbb{G})$ has large prime factors. This assumption seems to be heuristically reasonable if we look at parameters for CSIDH-$512$, for which the group $\mathbb{G}$ is isomorphic to $\mathbb{Z}/N\mathbb{Z}$ where $N$ is
\begin{align*}
    N =& \; 3 \cdot 37 \cdot 1407181 \cdot 51593604295295867744293584889 \\
     &\cdot 31599414504681995853008278745587832204909 \approx 2^{257.1},
\end{align*}
and its largest prime factor is roughly $2^{134.5}$ \cite{beullens2019csi}.
\begin{assumption}\label{assumption:groupaction-invertibility}
    Let $\{\mathbb{G}_\lambda\}_{\lambda \in \mathbb{N}}$ be a collection of finite abelian group with $\lambda$ being the security parameter such that the following is true. Suppose the prime factorization of the group order is $\ord(\mathbb{G}_\lambda) = p_1^{\alpha_1} \cdots p_k^{\alpha_k}$ where $p_1 \le \ldots \le p_k$. Then, it must hold that $p_k = \lambda^{\omega(1)}$.
\end{assumption}

\begin{lemma}\label{lemma:group-action-dTF}
Suppose the single-copy Extended Linear Hidden Shift Assumption holds over $(\mathbb{G}, \mathbb{X}, \star)$ with parameters $\lambda$ and $n \le \poly(\lambda)$. Further suppose that Assumption~\ref{assumption:groupaction-invertibility} holds for $\mathbb{G}$. Then, there exists an injective and invertible $\epsilon$-weak dual-mode trapdoor function family, where $\epsilon \le 1/\poly(\lambda)$.
\end{lemma}
By Lemma~\ref{lemma:dTF-amplification}, we immediately get the following corollary.

\begin{corollary}
Suppose the single-copy Extended Linear Hidden Shift Assumption holds over $(\mathbb{G}, \mathbb{X}, \star)$ with parameters $\lambda$ and $n \le \poly(\lambda)$, and that Assumption~\ref{assumption:groupaction-invertibility} holds for $\mathbb{G}$. Then, there exists an $\epsilon$-weak dual-mode trapdoor function family, where $\epsilon \le 1/\negl(\lambda)$.
\end{corollary}

\begin{proof}[Proof of Lemma~\ref{lemma:group-action-dTF}]
Let $B$ be an integer such that $B > 2n^2$ and $B = \poly(\lambda)$.\footnote{We have some flexibility in terms of setting the parameters: even a lower bound of just $B \ge 4n$ gives us a $1/2$-weak dTF, which we can amplify to get negligible correctness error using Lemma~\ref{lemma:dTF-amplification}.}

\paragraph{Efficient key generation.} The key-generation algorithm $\mathsf{Gen}(1^\lambda, \mu)$ works as follows.\\

\noindent
$\underline{\mathsf{Gen}(1^\lambda, \mu)}$:
\begin{enumerate}
    \item Sample $\matM \leftarrow \mathbb{G}^{n \times n}$, $\vecm \leftarrow \mathbb{G}^n$, $\vecs \leftarrow \{0,1\}^n$, and $\vect \leftarrow \mathbb{G}^n$.
    \item Sample $\vecu, \vecv \leftarrow \mathbb{G}^n$ and $\vecx_0 \leftarrow \mathbb{X}^n$ and set $\vecx_0, \vecx_1, \vecy_0, \vecy_1$ as follows:
    \begin{align*}
        \vecy_0 &:= \vect \star \vecx_0,\\
        \vecx_1 &:= \left[\matM \vecs + (1- \mu) \cdot \vecu \right] \star \vecx_0,\\
        \vecy_1 &:= \left[\matM \vecs + \vecm \odot \vecs + (1 - \mu) \cdot \vecv \right] \star \vecy_0.
    \end{align*}
    \item Set the key-trapdoor pair as
    \begin{align*}
        k &= (\matM, \vecm, \vecx_0, \vecy_0, \vecx_1, \vecy_1), \quad \text{and}\\
        t_k &= \Bigl(\vecs, \vect, (1 - \mu) \cdot (\vecu - \vecv)\Bigr).
    \end{align*}
\end{enumerate}

\paragraph{Efficient function evaluation.} Define a family of functions where for each $k, b$, we have that $f_{k, b} : [B]^n \rightarrow \mathbb{X}^{2n}$ is defined as 
\begin{align*}\label{equation:group-action-tcf}
f_{k, b}(\vecr) = \left(\vecz_{k, b}(\vecr), \vecz'_{k, b}(\vecr)\right),    
\end{align*} 
where
\begin{align*}
    \vecz_{k, b}(\vecr) &:= [\matM \vecr] \star \vecx_b, \quad \text{and}\\
    \vecz'_{k, b}(\vecr) &:= [\matM \vecr + \vecm \odot \vecr] \star \vecy_b \;.
\end{align*}
Assuming the group operation as well as the group action is efficiently computable, it is clear from the definition that the function can be efficiently evaluated give the key.

\paragraph{Efficient state preparation.} Let the distribution $D_{k, b}$ be the uniform distribution over $[B]^n$ for all $k, b$, so there is an efficient quantum procedure to prepare the state $\sum_{\vecr \in [B]^n} D_{k, b}(\vecr) \ket{\vecr}$.

\paragraph{Efficient inversion with trapdoor, and injectivity.} To invert the function $f_{k, b}$ on some output $(\vecz, \vecz')= f_{k, b}(\vecr)$, we can recover input $\vecr$ as follows. Observe that the following relation holds:
\begin{align*}
    \left[-\vect -b \cdot (1 - \mu) \cdot (\vecu - \vecv) \right] \star \vecz' = \left[ \vecm \odot (\vecr + b \cdot \vecs)\right] \star \vecz \;.
\end{align*}
Because the action on the right-hand side is applied component-wise and each entry of $\vecr$ can take on at most polynomially many values in $[B]$, we can recover each entry of $\vecr$ efficiently by brute force, since $\vecs$, $\vect$ and $(1 - \mu) \cdot (\vecu - \vecv)$ are included in the trapdoor. By Lemma~\ref{lemma:random-group-element-large-order}, we know that with all but negligible probability, every element of $\vecm$ has order larger than $B$, therefore, $\vecr$ is uniquely determined.

\paragraph{Dual-mode.} In the disjoint mode ($\mu = 0$) we claim that with all but negligible probability, the images of $f_{k, 0}$ and $f_{k, 1}$ are disjoint. Suppose for contradiction that there exist $\vecr, \vecr' \in [B]^n$ such that $f_{k, 0}(\vecr) = f_{k, 1}(\vecr')$. In particular, $z_{k, 0}(\vecr) = z_{k, 1}(\vecr')$. That is,
\begin{align*}
    [\matM \vecr] \star \vecx_0 &= [\matM \vecr'] \star \vecx_1 = [\matM (\vecr' + \vecs) + \vecu] \star \vecx_0.
\end{align*}
Since the group action is regular, we get that $\matM \vecr = \matM (\vecr' + \vecs) + \vecu$, or
\begin{align*}
    \matM (\vecr - \vecr' - \vecs) = \vecu.
\end{align*}
Varying $\vecv, \vecv', \vecs$, the left hand side can take at most $(2B+1)^n$ values. On the other hand, $\vecu$ is sampled uniformly at random from a set of size $|\mathbb{X}|^n$. Therefore, the probability over the choice of the key $k$ that there exist such values of $\vecr, \vecr'$ is at most $((2B+1)/|\mathbb{X}|)^n$, which is negligible in $\lambda$.

In the lossy mode ($\mu =1$), we recover (a simpler version of) the weak TCF of \cite{alamati2022candidate}. Observe that for any secret vector $\vecs \in \{0,1\}$, with probability at least $1 - n/(B-1)$ over the choice of $\vecr \leftarrow [B^n]$, $\vecr' := \vecr - \vecs$ is also in the domain $[B]^n$. By injectivity, this means that $|\alpha_{k, y, d}(0)| = |\alpha_{k, y, d}(1)|$ with the same probability. By the choice of parameters, $\epsilon = n/(B-1) \le 1/\poly(\lambda)$.

\paragraph{Mode indistinguishability.} In mode $\mu = 0$, since $\vecu$, $\vecv$ are sampled independently and uniformly, $\vecx_0, \vecx_1, \vecy_0, \vecy_1$ are i.i.d.~uniform random vectors, and therefore $k(0)$ follows the distribution $D_0$ in Equation~\eqref{equation:elhs-distributions}. On the other hand, in mode $\mu = 1$ the distribution of the key $k(1)$ is exactly $D_1$. Therefore the Extended Linear Hidden Shift assumption directly implies mode indistinguishability.
\end{proof}

\begin{lemma}\label{lemma:random-group-element-large-order}
    Suppose Assumption~\ref{assumption:groupaction-invertibility} holds for the collection of groups $\{\mathbb{G}_\lambda\}_{\lambda \in \mathbb{N}}$, and let $B = \poly(\lambda)$ be an integer. Then, for a random element $g \leftarrow \mathbb{G}$, we have that $\ord(g) \ge B$ with probability at least $1 - \negl(\lambda)$.
\end{lemma}

\begin{proof}
    Since $\mathbb{G}$ is a finite abelian group, we can decompose $\mathbb{G}$ as a direct sum of cyclic groups $A_{i, j}:= \mathbb{Z}_{p_i^{\beta_{i, j}}}$ of prime-power order, where $j \in [k_i]$, $\beta_{i, j} \ge 1$ and $\sum_j \beta_{i, j} = \alpha_i$.
    Sampling a uniformly random group element $g \leftarrow \mathbb{G}$ is equivalent to sampling its components $g_{i, j} \leftarrow A_{i, j}$ for all $i \in [k]$ and $j \in [k_i]$. 
    
    Observe that $\ord(g) \ge \max_{i, j} \ord(g_{i,j})$. In particular, consider a component cyclic group, say $A_{k, 1}$, whose order is a power of the largest prime $p_k$. Since every non-identity element of this component cyclic group has order at least $p_k$, with probability at least $1 - p_k^{-\beta_{k, 1}} = 1 - \lambda^{-\omega(1)}$, a random element $g_{k, 1} \leftarrow A_{k, 1}$ has order $\ord(g_{k, 1}) \ge p_k$. Since $B = \poly(\lambda) \le p_k$, this completes the proof.
\end{proof}

\paragraph{Acknowledgements.}
This research was supported in part by DARPA under Agreement Number HR00112020023, NSF CNS-2154149, a Simons Investigator award and a Thornton Family Faculty Research Innovation Fellowship. AG was supported in addition by the Ida M. Green Fellowship from MIT.

\bibliographystyle{alpha}
\bibliography{refs.bib}

\appendix

\ifnum\llncs=0
\section{The DSS Teleportation Gadget}\label{section:DSSgadget}
In this section, we describe the structure of the gadget and how the evaluator uses the gadget to correct the $\P$ gate errors introduced when performing a $\T$ gate.
As depicted in Figure~\ref{fig:DSS}, the procedure of running the DSS gadget takes as input a state $\rho$ and an encryption $\tilde{x} = \HE.\Enc_{pk}(x)$ for some $x \in \{0,1\}$ and outputs the state
\begin{align*}
    \X^{x'} \Z^{z'} (\P^\dagger)^x \rho (\P)^x \Z^{z'} \X^{x'},
\end{align*}
for some one-time pad keys $x', z' \in \{0,1\}$, along with their encryptions $\HE.\Enc_{pk'}(x')$ and $\HE.\Enc_{pk'}(z')$ under an independent public key $pk'$. At a high level, what the gadget decrypts $\tilde{x}$ and then applies $(\P^\dagger)^x$ to the input qubit in a way that does not reveal the secret key $sk$. Dulek et al.~\cite{dulek2016quantum} use the fact that the decryption circuit is in $\mathsf{NC}^1$, that is, its depth grows logarithmically with the security parameter, to write a polynomial-size width-$5$ permutation branching program $\mathcal{P}$ that computes $\HE.\Dec$ (via Barrington's theorem~\cite{barrington1986bounded}). The permutation branching program can then be encoded as Bell pairs that wired together in a particular manner. Teleporting an input qubit through these Bell pairs has the effect of performing the decryption computation.

\paragraph{The construction.} Now we present the construction of the DSS gadget explicitly, starting with the permutation branching program we want to encode. For intuition, see Example~\ref{example:dss-gadget-or} and Figure~\ref{fig:dss-gadget} taken from \cite{dulek2016quantum} for a picture of what the DSS gadget would look like if the decryption circuit was simply an OR of the first bits of the ciphertext and the secret key. Suppose that the decryption circuit $\mathsf{HE.Dec}$ takes as input the concatenation of the secret key $sk$ and the ciphertext $\tilde{x}$, and has depth $O(\log(\secp))$, where $\secp$ is the security parameter. Then, by Barrington's theorem, we can write $\mathsf{HE.Dec}$ as a width-$5$ permutation branching program $\mathcal{P}$ of length $L = \poly(\secp)$,
\begin{align*}
    \mathcal{P} = (\langle i_1, \sigma_1^0, \sigma_1^1 \rangle, \langle i_2, \sigma_2^0, \sigma_2^1 \rangle, \ldots, \langle i_L, \sigma_L^0, \sigma_L^1 \rangle).
\end{align*}
Without loss of generality, we can assume that $L$ is even, and that the instructions alternately depend on bits from the ciphertext $\tilde{x}$ and the secret $sk$. This can be ensured by inserting dummy instructions always performing the identity permutation between any two consecutive instructions that depend on the same variable (either $sk$ or $\tilde{x}$). Suppose that the $\ell$th instruction $\langle i_\ell, \sigma_\ell^1, \sigma_\ell^0 \rangle$ depends on $\tilde{x}$ if $\ell$ is odd and $sk$ if $\ell$ is even. 

The DSS gadget encodes the even-numbered instructions using entanglement between Bell pairs and Bell basis measurements to encode the odd-numbered instructions. The gadget is arranged in layers, and teleporting the input qubit through the DSS gadget implicitly runs $\HE.\Dec_{sk}$ on $\tilde{x}$, where each layer effectively applies a permutation in $S_5$ that teleports it to one of the five registers in the next layer.

\begin{enumerate}
    \item Let $Q$ be a set of $10L$ qubit registers. Label the first $5L$ of them in layers of ten as $1_{\ell, \text{in}}, 2_{\ell, \text{in}}, \ldots, 5_{\ell, \text{in}}$ and $1_{\ell, \text{out}}, 2_{\ell, \text{out}}, \ldots, 5_{\ell, \text{out}}$ for \textit{even} $\ell \in [2L]$.
    \item Let $Q_1$ be the set of registers labeled as $1_{\ell, \text{in}}, 2_{\ell, \text{in}}, \ldots, 5_{\ell, \text{in}}$ for even $\ell \le 2L$, and $Q_2$ be the set of registers labeled as $1_{\ell, \text{out}}, 2_{\ell, \text{out}}, \ldots, 5_{\ell, \text{out}}$ for \textit{even} $\ell \in [2L]$.
    \item Let $P = \{2_{L, \text{out}}, 3_{L, \text{out}}, \ldots, 5_{L, \text{out}}\}$ be the set of registers with the inverse phase gate applied to them.
    \item The Bell pairs between $Q_1$ and $Q_2$ are determined the even-numbered instructions $\langle i_\ell, \sigma_\ell^1, \sigma_\ell^0 \rangle$, for even $\ell \in [L]$.\footnote{For now we only describe the behavior of $\pi$ (and $\nu$) on the first $5L/2$ qubits of $Q_1$ (and $Q_2$).} In particular, if $a_{\ell, \text{in}} \in Q_1$ for some $a \in [5]$, then $\pi^{sk}(a_{\ell, \text{in}}) = \sigma_\ell(a)_{\ell, \text{out}}$, where $\sigma_\ell := \sigma_\ell^{sk_{i_\ell}}$ is determined by $i_\ell$th bit of $sk$.
    \item The Bell basis measurements to be performed by the evaluator depend on the ciphertext $\tilde{x}$ and are determined by the odd-numbered instructions $\langle i_\ell, \sigma_\ell^1, \sigma_\ell^0 \rangle$, for odd $\ell \in [L]$. If $a_{\ell-1, \text{out}} \in Q_2$ for some $a \in [5]$, then $\nu(a_{\ell, \text{out}}) = \sigma_\ell(a)_{\ell, \text{in}}$, where $\sigma_\ell := \sigma_\ell^{sk_{i_\ell}}$. Labeling the input qubit  as $1_{0, \text{out}}$, we also have that $\nu(1_{0, \text{out}}, \tilde{x}) = \sigma_1(1)_{1, \text{in}}$.
\end{enumerate}

By Barrington's theorem, if $\HE.\Dec_{sk}(\tilde{x}) = 0$ then the product, say $\tau$, of the permutations coming from the evaluated instructions equals the identity permutation. In that case, consecutively applying these permutations on `$1$' results in $\tau(1) = e(1) = 1$. The teleported qubit ends up in the register $1_{L, \text{out}}$ (with no inverse phase gate).
On the other hand, if $\HE.\Dec_{sk}(\tilde{x}) = 1$ then $\tau$ would evaluate to a $5$-cycle, and consecutively applying these permutations on `$1$' results in $\tau(1) \in \{2, 3, 4, 5\}$. In this case, the teleported qubit ends up in one of registers $2_{L, \text{out}}, 3_{L, \text{out}}, \ldots, 5_{L, \text{out}}$, with an inverse phase gate applied to it.

The gadget so far correctly applies the inverse phase gate $\P^\dagger$ if and only if $\HE.\Dec_{sk}(\tilde{x}) = x = 1$. However, the qubit is at a location that is unknown to the evaluator, because $\tau$ is hidden from them! To fix this problem, \cite{dulek2016quantum} use the remaining $5L$ qubits to execute the inverse branching program afterwards, so that the desired qubit is always teleported to the same register. The inverse branching program, which we denote as $\mathcal{P}^{-1}$ is given by the instructions of the original program $\mathcal{P}$ in reverse order, using the inverse permutations for each instruction instead,
\begin{align*}
    \mathcal{P}^{-1} = \Bigl(\langle i_L, (\sigma_L^1)^{-1}, (\sigma_L^0)^{-1} \rangle, \langle i_{L-1}, (\sigma_{L-1}^1)^{-1}, (\sigma_{L-1}^0)^{-1} \rangle, \ldots, \langle i_1, (\sigma_1^0)^{-1}, (\sigma_1^0)^{-1} \Bigr).
\end{align*}
After executing the entire gadget with $\mathcal{P}^{-1}$, the qubit is guaranteed to be in the location it started in: register `$1$' of the final layer of five qubits\footnote{Note that instruction $L$ is used twice in a row, breaking the alternation between `in' and `out' registers. This can be fixed by performing measurements that correspond to the identity permutation in between.}. This completes the description of the DSS gadget.

\begin{example}\label{example:dss-gadget-or}
As a simplified example, suppose that $sk, \tilde{x} \in \{0,1\}$ the decryption function $\HE.\Dec_{sk}(\tilde{a})$ is $sk \, \mathsf{OR} \, \tilde{x}$. Then, for one possible example set of values of $\tilde{a}$ and $sk$, half of the gadget and measurements will be as given in Figure~\ref{fig:dss-gadget}. To complete this gadget, the same construction is appended, reflected horizontally.
The $\mathsf{OR}$ function on two bits can be computed using a width-5 permutation branching program of length 4, consisting of the following list of instructions:
\begin{enumerate}
\item $\langle 1, e, (12345)\rangle$
\item $\langle 2, e, (12453)\rangle$
\item $\langle 1, e, (54321)\rangle$
\item $\langle 2, (14235), (15243)\rangle$
\end{enumerate}
\end{example}

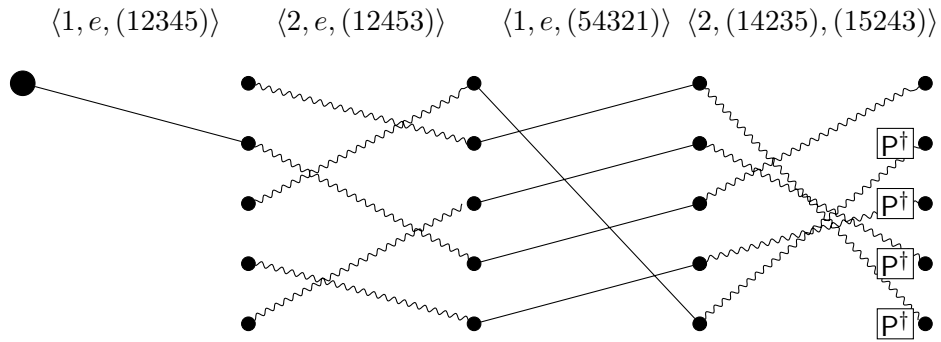
\begin{figure}
\centering
\begin{tikzpicture}[scale=0.8,every node/.style={draw,shape=circle,fill,scale=0.5},
decoration={snake, segment length=4, amplitude=1}]

\foreach \x in {2,...,5}
{
	\foreach \y in {1,...,5}
	{
		\pgfmathtruncatemacro{\label}{\x + 5 * \y}
		\node (\x-\y) at (3.75*\x,6-\y) {};
	}
}
\node[style={draw,shape=circle,fill,scale=1.8}] (1-1) at (3.75*1,6-1) {};

\draw     (1-1) -- (2-2);
          
\draw[decorate] (2-1) -- (3-2)
          (2-2) -- (3-4)
          (2-4) -- (3-5)
          (2-5) -- (3-3)
          (2-3) -- (3-1);

\draw     (3-5) -- (4-4)
          (3-4) -- (4-3)
          (3-3) -- (4-2)
          (3-2) -- (4-1)
          (3-1) -- (4-5);
 
 \draw[decorate]  (4-1) -- (5-5)
          (4-5) -- (5-2)
          (4-2) -- (5-4)
          (4-4) -- (5-3)
          (4-3) -- (5-1);

\filldraw[fill=white] (3.75*5 - 0.8,6-2 - 0.25) rectangle (3.75*5 - 0.2,6-2 + 0.25);
\node[draw=none,fill=none,scale=2] at (3.75*5 - 0.5, 6-2) {$\P^\dag$};
\filldraw[fill=white] (3.75*5 - 0.8,6-3 - 0.25) rectangle (3.75*5 - 0.2,6-3 + 0.25);
\node[draw=none,fill=none,scale=2] at (3.75*5 - 0.5, 6-3) {$\P^\dag$};
\filldraw[fill=white] (3.75*5 - 0.8,6-4 - 0.25) rectangle (3.75*5 - 0.2,6-4 + 0.25);
\node[draw=none,fill=none,scale=2] at (3.75*5 - 0.5, 6-4) {$\P^\dag$};
\filldraw[fill=white] (3.75*5 - 0.8,6-5 - 0.25) rectangle (3.75*5 - 0.2,6-5 + 0.25);
\node[draw=none,fill=none,scale=2] at (3.75*5 - 0.5, 6-5) {$\P^\dag$};

\node[draw=none,fill=none,scale=2] at (5.625,6) {$\langle 1, e, (12345)\rangle$};
\node[draw=none,fill=none,scale=2] at (9.375,6) {$\langle 2, e, (12453)\rangle$};
\node[draw=none,fill=none,scale=2] at (13.125,6) {$\langle 1, e, (54321)\rangle$};
\node[draw=none,fill=none,scale=2] at (16.875,6) {$\langle 2, (14235), (15243)\rangle$};

\end{tikzpicture}
\caption{(Taken from \cite{dulek2016quantum}, Figure~4): Structure of the (first half of the) gadget, with measurements, coming from the 5-permutation branching program for the $\mathsf{OR}$ function on the input $(sk, \tilde{x}) = (0,0)$. The example program's instructions are displayed above the permutations. The solid  lines correspond to Bell measurements, while the wavy lines
represent EPR pairs.}
\label{fig:dss-gadget}
\end{figure}
\fi

\end{document}